%% file: origins_V2.tex
\documentclass[a4paper,onecolumn,11pt,noarxiv]{quantumarticle}
\pdfoutput=1

\usepackage{mytikz}

\begin{document}

\title{On the system loophole of generalized non-contextuality}
\author{Victor Gitton}
\email{vgitton@ethz.ch}
\affiliation{Institute for Theoretical Physics, ETH Zürich, Switzerland}
\author{Mischa P.\ Woods}
\email{mischa.woods@gmail.com}
\affiliation{Institute for Theoretical Physics, ETH Zürich, Switzerland}
\affiliation{University Grenoble Alpes, Inria, Grenoble, France}
\maketitle

\begin{abstract}
Generalized noncontextuality is a well-studied notion of classicality that is applicable to a single system, as opposed to Bell locality. 
It relies on representing operationally indistinguishable procedures identically in an ontological model.
However, operational indistinguishability depends on the set of operations that one may use to distinguish two procedures: we refer to this set as the reference of indistinguishability.
Thus, whether or not a given experiment is noncontextual depends on the choice of reference.
The choices of references appearing in the literature are seldom discussed, but typically relate to a notion of system underlying the experiment.
This shift in perspective then begs the question: how should one define the extent of the system underlying an experiment?
Our paper primarily aims at exposing this question rather than providing a definitive answer to it.
We start by formulating a notion of relative noncontextuality for prepare-and-measure scenarios, which is simply noncontextuality with respect to an explicit reference of indistinguishability.
We investigate how verdicts of relative noncontextuality depend on this choice of reference, and in the process introduce the concept of the noncontextuality graph of a prepare-and-measure scenario. 
We then discuss several proposals that one may appeal to in order to fix the reference to a specific choice,
and relate these proposals to different conceptions of what a system really is.
With this discussion, we advocate that whether or not an experiment is noncontextual is not as absolute as often perceived.
\end{abstract}

\tableofcontents

\section{Introduction}

\paragraph{Non-classicality in quantum theory.}

On the face of it, quantum theory looks radically different from classical mechanics.
However, how does one capture this difference formally?
The seminal work of Bell \cite{bell_einstein_1964} demonstrated the stringent difference between correlations obtained across spacelike-separated agent depending on whether the resources they share are classical or quantum.
This gap is at the heart of many important quantum protocols, such as quantum key distribution \cite{ekert_quantum_1991}, quantum communication \cite{buhrman_nonlocality_2010}, or randomness certification \cite{pironio_random_2010} (see also \citeref{brunner_bell_2014} for a review of Bell nonlocality).
The settings of such protocols inevitably involve several agents that are macroscopically separated.
Yet, our intuition is that even a single-system quantum experiment defies our classical intuition.
Kochen and Specker \cite{kochen_problem_1967} (see \citeref{peres_incompatible_1990} for a simplified argument) paved the way of devising formal requirements that a hidden variable model should satisfy in the context of a single system.
The crucial idea therein is that joint measurability of quantum observables should be reflected at the level of the hidden variable model in a noncontextual manner, i.e., in a manner where the outcome of an observable $A$ does not depend on the context of whether one jointly measures the compatible observable $B$ or $C$ together with $A$.
This notion of noncontextuality has been shown to be a resource for quantum computation \cite{raussendorf_contextuality_2013,howard_contextuality_2014,bermejo-vega_contextuality_2017,frembs_contextuality_2018} and quantum communication \cite{saha_state_2019,gupta_quantum_2022}.

\paragraph{Generalized noncontextuality.}

Spekkens \cite{spekkens_contextuality_2005} conceptually generalized this notion to what is now known as generalized noncontextuality. 
Generalized noncontextuality is a constraint that is built on top of hidden variable models, and in this context, one would rather think of hidden variable models as being tentative ontological models of nature \cite{hardy_quantum_2004,spekkens_contextuality_2005}. An ontological model ought to describe what is really going on in a system, and the meaning of ``what is really going on'' is acquired through the constraints that one implements on top of the ontological model.
Specifically, the constraints of generalized noncontextuality follow from the principle that operational equivalences should be explained by equivalences within the ontological model.
This principle has been traced back to Leibniz and identified as a crucial methodological principle for the construction of general relativity in \citeref{spekkens_ontological_2019}.
More practically, generalized noncontextuality has been demonstrated to be a resource in parity-oblivious random access coding \cite{spekkens_preparation_2009,chailloux_optimal_2016,ambainis_parity_2019}, state discrimination \cite{schmid_contextual_2018,mukherjee_discriminating_2021,flatt_contextual_2021,shin_quantum_2021}, state-dependent cloning \cite{lostaglio_contextual_2020} and quantum communication \cite{yadavalli_contextuality_2021}.

\begin{figure}[ht]
\input{venn_diagram}
\caption{Possible approaches to noncontextuality given a prepare-and-measure scenario (``NC'' stands for ''noncontextuality''). 
The operational procedures (i.e., preparations and measurements) of the prepare-and-measure scenario are enclosed in the shaded blue region. 
The regions enclosed by the red circles correspond to the set of procedures that should be used as a reference of indistinguishability in noncontextuality according to different approaches.
In operational noncontextuality, one chooses the procedures of the prepare-and-measure scenario as the reference. 
Pragmatic noncontextuality allows for investigating several choices of reference that extend beyond the prepare-and-measure scenario, and each choice of reference is thought of as being more or less relevant according to some pragmatic considerations. 
In ontic-system noncontextuality, the prepare-and-measure scenario is thought of as being a subset of procedures that probe the true extent of the ontic system's set of operations, although the latter is in general unknown.}
\label{fig:venn}
\end{figure}
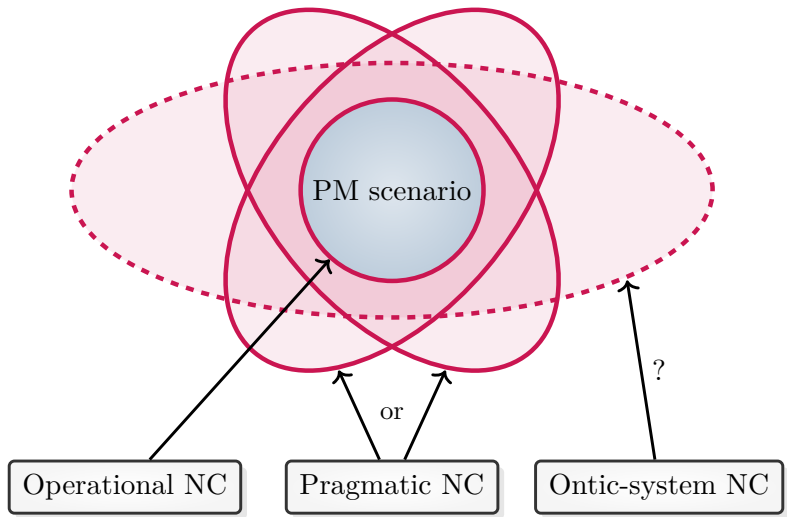

\paragraph{Objectives of this work.}

The constraints of generalized noncontextuality require a notion of equivalence between operational procedures.
However, the exact methodology that one should follow to identify the relevant operational equivalences that hold in a given experiment often evades a precise account in the existing literature.
Interrogating this methodology has long-reaching implications. For instance, a primary object of study in this field is the existence of a generalized noncontextual ontological model for a prepare-and-measure scenario, and this existence depends on the operational equivalences that one takes into account when formulating the noncontextuality constraints.
The objective that we set for this work is the following. 
We would like to attempt to re-build a notion of generalized noncontextuality in a way that treats every assumption of the formalism in an operational manner when possible, and that otherwise discusses carefully the remaining conceptual aspects.
This leads us to formulate a notion of relative noncontextuality that explicitly features a choice of reference procedures that one uses to determine the operational equivalences that hold in the prepare-and-measure scenario, and to then discuss different proposals for how this choice can be fixed.
Each such proposal yields a different notion of noncontextuality, and some of these proposals are related to different conceptions of what a system really is: this is sketched in \cref{fig:venn}.
These proposals attempt to reflect the existing literature when possible.
As we shall argue, these proposals each have limitations with respect to the extent to which they are well-defined in their current stage or with respect to how they fit in the existing literature.
While we do not claim to be able to settle this discussion, we would like to bring it forward in the community. Furthermore, we would like to advocate in favor of more transparency with respect to the methodology that one chooses to endorse in a given study and more discussions with respect to the motivations and foundations of such methodologies.

\begin{figure}[ht]
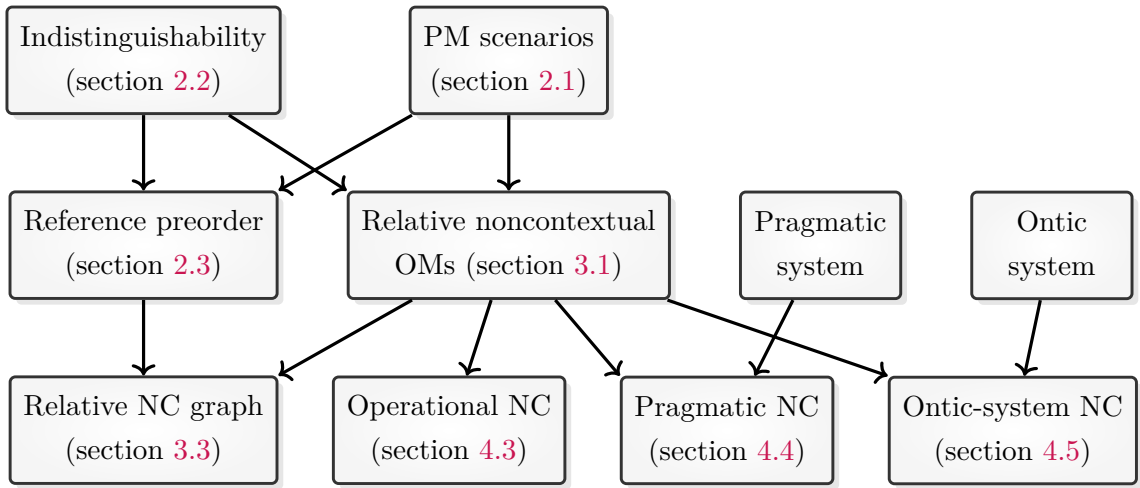

\include{mind_map}
\caption{The relationships between the main notions of this work. For brevity, we use ``PM'' for ``prepare-and-measure'', ``OMs'' for ``ontological models'', and ``NC'' for ``noncontextuality''. A directed edge indicates that the child notion builds on top on the parent notion.}
\label{fig:main notions}
\end{figure}

\paragraph{Outline of the paper.}

The main notions are summarized in \cref{fig:main notions}.
We start in \cref{sec:pm scenarios} by introducing the necessary operational ingredients. The formalism that we use is mostly standard in the literature and follows closely that of, e.g., \citeref{spekkens_contextuality_2005}, but we make sure not to presume an underlying notion of system. Prepare-and-measure scenarios are introduced in \cref{sec:operational primitices},
and our notion of indistinguishability is introduced in \cref{sec:introducing indistinguishability}: an important concept is that of the reference of indistinguishability, which describes the set of procedures that one uses in order to establish indistinguishability relations.
In \cref{sec:reference preorder}, given a prepare-and-measure scenario, we introduce a preorder between different choices of references: the preorder describes the case where a reference is a fine-graining of another one with respect to the indistinguishability relations that are seen at the level of the prepare-and-measure scenario.

\Cref{sec:relative noncontextuality} focuses on the notion of relative noncontextuality. 
In \cref{sec:definitions om}, we formulate at first a definition of an ontological model for a prepare-and-measure scenario, and then specialize to relative noncontextual ontological models in \cref{def:classical model}.
A relative noncontextual ontological model can be thought of as a generalized noncontextual ontological model in the original sense \cite{spekkens_contextuality_2005}, but with an explicit parametrization of the indistinguishability relations that one imposes at the ontological level. In that sense, formally speaking, a relative noncontextual ontological model is similar to the generalized noncontextual ontological model of \citeref{schmid_all_2018}.
The key difference is that while \citeref{schmid_all_2018} assumed a fixed set of operational equivalences as an input, without necessarily specifying how these should be obtained in practice, we parametrize this set of operational equivalences by a choice of reference procedures, which is quite natural and convenient for our purposes. Indeed, the reference procedures will later be related to the procedures that one considers as probing the system.
Then, \cref{sec:varying the reference} formulates some basic results that relate the reference preorder to the existence of relative noncontextual ontological models for different references.
This leads, in \cref{sec:relative nc graph}, to the notion of relative noncontextuality graph: given a prepare-and-measure scenario, one introduces a vertex in the graph for each choice of reference and indicates, for each vertex, whether there exists a relative noncontextual ontological model for the prepare-and-measure scenario with respect to the reference assigned to the vertex. 
The vertices of the graph are then ordered with respect to the reference preorder, which constrains the structure of the relative noncontextuality graph according to the results of \cref{sec:varying the reference}.
As shown in \cref{fig:typical graph}, the generic relative noncontextuality graph will display a behavior where noncontextuality emerges as one fine-grains the indistinguishability relations of the prepare-and-measure scenario.

In \cref{sec:choices of references}, we investigate possible approaches with respect to fixing the reference of indistinguishability that one uses in relative noncontextuality.
This is motivated by the fact that generalized noncontextuality in its standard form appears to be concerned with a fixed set of operational equivalences \cite{spekkens_contextuality_2005,schmid_all_2018,spekkens_ontological_2019}.
We first discuss the notion of in-principle indistinguishability \cite{spekkens_ontological_2019,selby_open-source_2022,catani_reply_2022} in \cref{sec:limits of in-principle indistinguishability}, and argue in favor of replacing it, in this context, with a notion of indistinguishability with respect to a set of procedures, or indistinguishability with respect to all procedures on a system: indeed, we argue that in-principle indistinguishability is always fundamentally of this form.
In \cref{sec:operational noncontextuality}, we describe the proposal of operational noncontextuality \cite{gitton_solvable_2022} that chooses the extent of the prepare-and-measure scenario itself to define the reference of indistinguishability, as represented in \cref{fig:venn}.
This proposal has the merit of being methodologically simple, but perhaps at the expense of being agnostic of more general considerations.
The proposal of \cref{sec:pragmatic noncontextuality} is that of pragmatic noncontextuality: a pragmatic system is a system defined by a set of procedures, and this system is argued to be more or less relevant based on pragmatic considerations (e.g., how interesting that system is, how it is related to interesting tasks, how fundamental it is believed to be, etc.).
In the approach of pragmatic noncontextuality, one then embeds the prepare-and-measure scenario into different pragmatic systems, as shown in \cref{fig:venn}, each time choosing the reference of indistinguishability to be used in noncontextuality as corresponding to the procedures that define the pragmatic system.
Verdicts of pragmatic contextuality are then weighted by the relevance of the choice of pragmatic system under consideration.
In \cref{sec:ontic-system noncontextuality}, we discuss the notion of ontic systems, which are conceptual entities whose operational content is in general unknown as suggested in \cref{fig:venn}. This results in the notion of ontic-system noncontextuality: there, given a prepare-and-measure scenario, one should have a prescription of the ontic system that the prepare-and-measure scenario is interacting with, and choose the true tomographically complete set of operations that probe the ontic system as the reference of indistinguishability: this is a possible reading of \citeref{pusey_contextuality_2019}.
Finally, in \cref{sec:local noncontextuality}, we discuss the limits of replacing pragmatic systems and ontic systems with locality considerations.
We conclude in \cref{sec:conclusion} with some open questions.

\newpage
\section{Prepare-and-measure scenarios and operational indistinguishability}
\label{sec:pm scenarios}

In this section, we introduce the operational framework on top of which we describe the type of prepare-and-measure scenarios that we shall consider. Furthermore, we introduce a notion of operational indistinguishability that will become crucial in \cref{sec:relative noncontextuality} when we define relative noncontextual ontological models.

\subsection{Operational framework}
\label{sec:operational primitices}

\paragraph{Operational procedures and events.} We consider an agent in a lab. The agent has access to a number of possible actions (``press the blue button'', ``turn the red knob to the right'', etc.).
The agent furthermore has access to a measurable space $\measspace$, where $\eventset$ is a set of atomic observation events (``the green light is flashing'', ``the Geiger counter ticked'', ``the coin landed tails'', etc.) and $\sigmaalg$ is a $\sigma$-algebra over $\eventset$, i.e., a collection of subsets of $\eventset$ that is closed under complement, and countable unions and intersections such that $\emptyset,\eventset\in\sigmaalg$. Here, $\emptyset$ is the empty set that describes the ``null event'' that never occurs, while $\eventset$ corresponds to the ``trivial event'' that always occurs. An event that can be observed by the agent is any $\omega \in \sigmaalg$. For instance, $\omega = \{\text{``the green light is flashing''}, \text{``the coin landed tails''}\}$ is the event describing that the green light is flashing or the coin landed tails (or both). 

\paragraph{Procedures, preparations and effects.} 
The agent can then generate what we refer to as operational procedures, typically denoted $\proc$, by chaining certain actions, conditioned on the observed events that happened in the past. The set of such procedures is denoted $\procset$. This tacitly requires the agent to be equipped with a trusted classical memory. In other words, a procedure is a finite algorithm of the form ``execute this action first, then, conditioned on observing this event, execute this second action, else execute this third action, etc.''.
Preparations, typically denoted $\prep \in \procset$ (where the letter choice comes from the analogy to ``states'')
are simply instances of such general operational procedures.
On the other hand, an effect, typically denoted $\effect$, consists of a procedure $\proc\in\procset$ together with an event $\event\in\sigmaalg$ that describes the condition of ``success'' of the effect. This is typically denoted $\effect = \event|\proc \in \effectset = \sigmaalg \times \procset$.
For instance, the effect $\effect$ could correspond to ``press the blue button and consider the effect a success if the red light is on afterwards'', which is denoted $\effect = \effectexpl$ where $\event\in\sigmaalg$ corresponds to ``the red light is on'' and $\proc\in\procset$ corresponds to ``press the blue button''.
In general, one can compose procedures together, but for our purposes we shall only look at composing effects after preparations. The result is also an effect: for instance, the effect $\effect|\prep$ 
is thought of as being successful if, given a decomposition $\effect = \effectexpl$, the agent sees the event $\event$ after carrying out the preparation $\prep$ followed by the procedure $\proc$.
We will occasionally consider effects that correspond to observing an event $\event\in\sigmaalg$ without any prior action: this could be denoted $\event|\text{``do nothing''} \in \effectset$, but we will simply denote this $\event \in \effectset$ by a slight abuse of notation.

\paragraph{Probabilities.}
For all $\effect\in\effectset$, for all $\prep\in\procset$, we assume that the agent has a prescribed probability $\probarg{\effect}{\prep} \in [0,1]$ that describes the trust that the agent has in the success of the procedure $\effect|\prep$. These probabilities may for instance come from an assumption that the agent can repeat the procedure $\effect|\prep$ in an independent manner many times, and the probability $\probarg{\effect}{\prep}$ is the best estimate known for the expected frequency of success of the procedure $\effect|\prep$, but this does not really matter for our purposes.
Still, this probability distribution must satisfy the following requirements: for all preparations $\prep \in \procset$,
\begin{itemize}
\item[---] \textbf{trivial event $\eventset$:} for all $\effect\in\effectset$ such that $\effect = \eventset|\proc$ for some $\proc\in\procset$, it holds that $\probarg{\effect}{\prep} = 1$;
\item[---] \textbf{additivity:} for all $\effect_1,\effect_2,\effect_3\in\effectset$ such that there exist $\proc\in\procset$ and $\event_1,\event_2 \in \sigmaalg$ satisfying
\begin{subequations}
\label{eq:additivity case}
\begin{align}
    \effect_1 &= \event_1 | \proc, \\
    \effect_2 &= \event_2 | \proc, \\
    \effect_3 &= \event_1\cup\event_2 |\proc, &\textup{\small($\cup$ denotes set union)}\\
    \emptyset &= \event_1\cap\event_2, &\textup{\small($\cap$ denotes set intersection)}
\end{align}
\end{subequations}
it holds that $\probarg{\effect_1}{\prep} + \probarg{\effect_2}{\prep} = \probarg{\effect_3}{\prep}$. In other words, if an effect is a coarse-graining of mutually incompatible effects, then the probabilities are additive.
\end{itemize}
These conditions together with the trivial identity $\eventset = \eventset \cup \emptyset$ imply in particular that if $\effect = \emptyset|\proc$ for some $\proc \in \procset$, then $\probarg{\effect}{\prep} = 0$ for all $\prep \in \procset$.

\paragraph{Classical mixtures.}

We additionally endow the agent with a trusted classical random number generator. This allows the agent to generate classical mixtures of preparations, referred to as preparation densities, and classical mixtures of effects, referred to as effect densities.
For instance, given a set of preparations $\sets \subseteq \procset$ and a set of effects $\sete \subseteq \effectset$, the set of preparation densities is denoted $\convhull\sets$ and the set of effect densities is denoted $\convhull\sete$. A preparation density $\prepdens \in \convhull\sets$ stands for the formal convex mixture $\prepdens = \sum_{\prep\in\sets} \prepdensarg\prep \prep$, 
where the real coefficients $\{\prepdensarg\prep\}_{\prep\in\sets}$ are such that $\prepdensarg\prep \geq 0$ and $\sum_{\prep\in\sets} \prepdensarg\prep = 1$. Similarly, an effect density $\effectdens \in \convhull\sete$ stands for the formal convex mixture $\effectdens = \sum_{\effect\in\sete} \effectdensarg\effect \effect$ where the real coefficients $\{\effectdensarg\effect\}_{\effect\in\sete}$ are such that $\effectdensarg\effect \geq 0$ and $\sum_{\effect\in\sete} \effectdensarg\effect = 1$. 
The resulting probabilities of success are then obtained in a convex-linear fashion:
\begin{equation}
    \probarg{\effectdens}{\prepdens} = \sum_{\effect\in\sete,\prep\in\sets} \effectdensarg\effect\prepdensarg\prep \probarg{\effect}{\prep},
\end{equation}
i.e., the function $\probarg{\cdot}{\cdot} : \effectset \times \procset \to [0,1]$ is convex-linear in both arguments.

\begin{figure}[ht]
\centering
$$
\drawdatatable{\drawexampletable}
$$
\caption{An instance of a data table. Each entry represents the probability $\probarg{\effect}{\prep} \in [0,1]$ of success that one attributes to the execution of the preparation $\prep$ associated to each column followed by the execution of the effect $\effect$ associated to each row. The effects $\emptyset$ and $\eventset$ correspond to the null and trivial events, respectively.}
\label{fig:example datatable}
\end{figure}

\paragraph{Data tables.}

It will be convenient to represent the probabilities associated to a set of preparations and effects in the form of a data table. A typical data table is shown in \cref{fig:example datatable}.
This specific data table is rather idealized since it contains $1$'s and $0$'s for non-trivial effects. More generally, a data table can be denoted abstractly as a matrix $\{\probarg{\effect}{\prep}\}_{\effect\in \sete,\prep\in \sets}$. The sets $\sets\subseteq\procset$ of preparations and $\sete\subseteq\effectset$ of effects are sets of interests.
Importantly, we do not really need nor want to assume that a data table is ever ``complete'': in principle, there are always additional measurements one could add in the rows and preparations in the columns (which may or may not be redundant). For our purposes, what is important is to remind ourselves that every data table could in principle be extended, but that we nonetheless restrict our attention to a particular finite data table.
We shall always assume that any preparation $\prep$ or effect $\effect$ that we mention is such that we know the associated probability $\probarg{\effect}{\prep}$.

\paragraph{Prepare-and-measure scenarios.}

In the rest of this work, we will use this definition:
\begin{definition}
\label{def:pm scenario}
A prepare-and-measure scenario $\apm$ consists of a finite set of preparations $\asets \subset \procset$ and a finite set of effects $\asete \subset \effectset$ such that the null event and trivial event belong to $\asete$, i.e., $\emptyset,\eventset \in \asete$.
\end{definition}
In principle, one could allow for infinite sets $\asets$ and $\asete$, but this is not very important for our purposes.
As we will later clarify, given a prepare-and-measure scenario $\apm$, one is really interested in preparation densities belonging to $\convhull\asets$ and effect densities belonging to $\convhull\asete$.
Typically, the prepare-and-measure scenario is a subset of a more general set of preparations and effects that one has performed. 
Thus, typically, a prepare-and-measure scenario's data table is a subset of the data table that one has investigated. We can represent this graphically: for instance, we can pick $\asets = \{\prep_i\}_{i=1}^3$ and $\asete = \{\emptyset,\eventset\} \cup \{\effect_i\}_{i=1}^2$ in the example of \cref{fig:example datatable}, which we represent as in \cref{fig:example pm scenario}.
\begin{figure}[ht]
\centering
$$
\drawdatatable{
\drawexamplepmscenario
}
$$
\caption{An example prepare-and-measure scenario: the set of preparations $\asets$ (set of effects $\asete$) belonging to the prepare-and-measure scenario is a subset of the columns (rows), and we shade the relevant probabilities.}
\label{fig:example pm scenario}
\end{figure}

\subsection{Indistinguishability}
\label{sec:introducing indistinguishability}

Noncontextual ontological models are based on relations of operational equivalence, or equivalently, of indistinguishability. We make these notions precise here.
Let us define what a reference of indistinguishability, or, for brevity, a reference, is:
\begin{definition}[Reference of indistinguishability]
\label{def:ref}
A reference (of indistinguishability) is any pair $\rpm$ consisting of a set of preparations $\rsets \subseteq \procset$, the reference preparations, and a set of effects $\rsete \subseteq \effectset$, the reference effects.
\end{definition}
 Let $\apm$ be a prepare-and-measure scenario. 
The role of the reference effects in $\rsete$ is to define an equivalence relation of indistinguishability on preparation densities in $\convhull\asets$: we say that two preparation densities $\prepdens,\prepdens'\in\convhull\asets$ are indistinguishable with respect to $\rsete$, denoted $\prepdens \indist{\rsete} \prepdens'$, if and only if it holds that
\begin{equation}
\forall \effect\in\rsete\st \probarg{\effect}{\prepdens} = \probarg{\effect}{\prepdens'}.
\end{equation}
The binary relation $\indist{\rsete}$ is an equivalence relation, since it is clearly reflexive, symmetric and transitive. 
We repeat these considerations but exchanging the role of preparations and effects: two effect densities $\effectdens,\effectdens' \in \convhull\asete$ are indistinguishable with respect to $\rsets$, denoted $\effectdens \indist{\rsets} \effectdens'$, if and only if it holds that
\begin{equation}
\forall \prep \in \rsets \st \probarg{\effectdens}{\prep} = \probarg{\effectdens'}{\prep}.
\end{equation}
\begin{figure}[!t]
\centering
\begin{subfigure}{0.8\textwidth}
\centering
$$
\drawdatatable{
\drawindistdatatable
%
\shaderegion{red}{(1,1)}{(4,5)}
\draw[bracestyle,red] (4.25,-1.5) -- node[left,black] {$\rsete^{(1)}$} (0.75,-1.5);
%
\begin{scope}[line width=1pt,black]
\draw (0.6,2.75) rectangle (4.4,3.25);
\draw (0.6,4.25) rectangle (4.4,4.75);
\node[black] (equive) at (6.5,3.75) {$\prep_3\indist{\rsete^{(1)}}\frac12\prep_1+\frac12\prep_2$};
\draw[arrows={->}] (4.4,3.25) -- (equive);
\draw[arrows={->}] (4.4,4.25) -- (equive);
\end{scope}
}
$$
\caption{An example indistinguishability relation.}
\label{fig:example indist 1}
\end{subfigure}
\begin{subfigure}{0.8\textwidth}
\centering
$$
\drawdatatable{
\drawindistdatatable
%
\shaderegion{red}{(1,1)}{(5,5)}
\draw[bracestyle,red] (5.25,-1.5) -- node[left,black] {$\rsete^{(2)}$} (0.75,-1.5);
\shaderegion{blue}{(1.1,1.1)}{(3.9,4.9)}
%
\begin{scope}[line width=1pt,black]
\draw (0.6,2.75) rectangle (5.4,3.25);
\draw (0.6,4.25) rectangle (5.4,4.75);
\node[black] (equive) at (7,3.75) {$\prep_3\nindist{\rsete^{(2)}}\frac12\prep_1+\frac12\prep_2$};
\draw[arrows={->}] (5.4,3) -- (equive);
\draw[arrows={->}] (5.4,4.5) -- (equive);
\end{scope}
}
$$
\caption{The indistinguishability relation of \cref{fig:example indist 1} no longer holds because of the larger reference $\rsete^{(2)}$.}
\label{fig:example indist 2}
\end{subfigure}
\caption{Indistinguishability relations. Visually, looking for a preparation indistinguishability relation in a data table amounts to looking for two convex combinations of columns such that the components \emph{within the subset of rows corresponding to the reference $\rsete$} agree. Looking for an effect indistinguishability relation is analogous, but exchanging the roles of rows and columns.}
\label{fig:example indist}
\end{figure}

\paragraph{Example.}

Consider the data table of \cref{fig:example datatable}. We choose the prepare-and-measure scenario $\asets = \{\prep_i\}_{i=1}^3$ and $\asete = \{\emptyset,\eventset\}\cup \{\effect_i\}_{i=1}^2$, and the reference sets $\rsets = \asets$ and $\rsete^{(1)} = \asete$. This yields the indistinguishability relation
\begin{equation}
\label{eq:example indist 1}
\prep_3 \indist{\rsete^{(1)}} \frac12 \prep_1 + \frac12 \prep_2.
\end{equation}
On the other hand, choosing $\rsete^{(2)} = \{\emptyset,\eventset\}\cup\{\effect_i\}_{i=1}^3$ lifts this indistinguishability relation:
\begin{equation}
\label{eq:example indist 2}
\prep_3 \nindist{\rsete^{(2)}} \frac12 \prep_1 + \frac12 \prep_2,
\end{equation}
since these two preparation densities yield different probabilities for the effect $\effect_3\in\rsete^{(2)}$.
These considerations are summarized in \cref{fig:example indist}: the setting of \cref{eq:example indist 1} is represented in \cref{fig:example indist 1},
while the setting of \cref{eq:example indist 2} is represented in \cref{fig:example indist 2}.

\subsection{Reference preorder and faithful references}
\label{sec:reference preorder}

Let us define the preorder relation $\rpm\refordereq \rpmm$ that intuitively describes that $\rpmm$ is a ``better or equivalent'' reference than $\rpm$ with respect to distinguishing the procedures of the prepare-and-measure scenario $\apm$.
In the following, a preorder $\preceq$ is a relation on a set $X$ such that $x \preceq x$ for all $x \in X$ (reflexivity) and such that $x \preceq x'$ and $x'\preceq x''$ implies $x\preceq x''$ (transitivity).\footnote{In contract, a partial order is such that $x \preceq x'$ and $x' \preceq x$ implies $x = x'$: in our case, there can be different choices of references that induce the same indistinguishability relations, so we do not have a partial order to start with.}

\newpage
\begin{definition}[Reference preorder]
\label{def:ref preorder}
Let $\apm$ be a prepare-and-measure scenario, and let $\rpm$ and $\rpmm$ be two choices of references. We define the preorder $\refordereq$ such that
\begin{equation}
\rpm \refordereq \rpmm
\end{equation}
if and only if, for all $\prepdens,\prepdens'\in\convhull\asets$, for all $\effectdens,\effectdens'\in\convhull\asete$, the following implications hold:
\begin{subequations}
\label{eq:preorder def}
\begin{align}
    \prepdens \indist{\rsete'} \prepdens' &\limplies \prepdens \indist{\rsete} \prepdens', \\
    \effectdens \indist{\rsets'} \effectdens' &\limplies \effectdens \indist{\rsets} \effectdens'.
\end{align}
\end{subequations}
\end{definition}

Not all choices of references $\rpm$ are satisfactory for our later endeavors. 
A faithful reference is one that allows to better (or equally well) distinguish
the preparations and effects of a prepare-and-measure scenario as the prepare-and-measure scenario itself does. 
This can be guaranteed by ensuring that $\asets \subseteq \rsets$ and $\asete \subseteq \rsete$, but this condition is not necessary.
Indeed, it could be that subsets $\rsets \subseteq \asets$ and $\rsete\subseteq\asete$ still give rise to a faithful reference $\rpm$ for $\apm$: this could be the case if, in some loose sense, $\rpm$ is a ``tomographically complete'' subset of $\apm$.
The criterion that will become relevant is the following:%
\begin{definition}[Faithful reference]
\label{def:faithful reference}
Let $\apm$ be a prepare-and-measure scenario. A reference $\rpm$ is a faithful reference with respect to $\apm$ if and only if $\apm \refordereq \rpm$ with respect to the preorder of \cref{def:ref preorder}.
\end{definition}
As we will later point out, choosing an unfaithful reference for a prepare-and-measure scenario forbids the existence of a relative noncontextual ontological model.

\newpage
\section{Relative noncontextuality }
\label{sec:relative noncontextuality}

In this section, building upon the operational formalism of \cref{sec:pm scenarios}, we will introduce relative noncontextual ontological models. These should be thought of as being a conceptual generalization of the noncontextual ontological models of \citeref{spekkens_contextuality_2005} in the sense that they allow to explicitly parametrize the indistinguishability relations that one imposes at the level of the ontological model through the choice of a reference $\rpm$.
In \cref{sec:choices of references}, we will attempt to describe how existing notions of noncontextuality can be recovered from relative noncontextuality together with a prescription of how to fix the reference $\rpm$.

\subsection{Definitions}
\label{sec:definitions om}

\paragraph{Ontological models.} 

Let us first introduce ontological models. An ontological model \cite{hardy_quantum_2004,spekkens_contextuality_2005} for a prepare-and-measure scenario $\apm$ (\cref{def:pm scenario}) is a causal model that mediates the influence of the preparation to the success of the effect through a classical hidden variable. More specifically, it consists of:
\begin{myitem}
\item a finite\footnote{In some contexts \cite{schmid_all_2018,gitton_solvable_2022}, noncontextuality assumptions can be used to explicitly prove that one can assume $\ospace$ to be finite without loss of generality compared to the generic infinite case. In this work, we assume a finite ontic set $\ospace$ for simplicity.} set $\ospace$;
\item a probability distribution $\ostate{\prepdens}$ over $\ospace$ for each preparation density $\prepdens \in \convhull\asets$;
\item an effect $\orep{\effectdens}$ over $\ospace$ for each effect density $\effectdens \in \convhull\asete$.
\end{myitem}
We refer to $\ospace$ as an ontic space, to $\ostate\prepdens$ as an ontic state distribution, and to $\orep\effectdens$ as an ontic response function. 
The intuition of such a model is that upon implementing the preparation density $\prepdens \in \convhull\asets$, one really prepares ``something'' in the state $\hv \in \ospace$ with probability $\ostate{\prepdens}(\hv)$,
and if a given state $\hv\in\ospace$ was ``really prepared'', then one would observe the success of the effect density $\effectdens\in\convhull\asete$ with probability $\orep{\effectdens}(\hv)$ --- one thus expects that $\probarg{\effectdens}{\prepdens} = \sum_{\hv\in\ospace} \orep{\effectdens}(\hv) \ostate{\prepdens}(\hv)$, on average.
The consistency conditions that the ontological model must satisfy are made explicit in \cref{def:ontological model}:
\begin{definition}[Ontological model]
\label{def:ontological model}
Consider a prepare-and-measure scenario $\apm$. 
An ontological model for $\apm$ is a triplet 
\begin{equation}
\omodel
\end{equation}
such that $\ospace$ is a finite set, and for all $\hv\in\ospace$, for all $\prepdens,\prepdens' \in \convhull\asets$, for all $\effectdens,\effectdens' \in \convhull\asete$, the following holds:
\begin{myitem}
\item the ontic state distributions are well-defined: 
\begin{equation}
\label{eq:ontic state well-defined}
    \ostate{\prepdens}(\hv) \geq 0 \quad \textup{and} \quad \textstyle\sum_{\hv\in\ospace} \ostate{\prepdens}(\hv) = 1;
\end{equation}
\item the ontic response functions are well-defined:
\begin{equation}
\label{eq:ontic response function well-defined}
0 \leq \orep{\effectdens}(\lambda) \leq 1;
\end{equation}
\item the ontic response function implements the trivial event trivially: for all $\effect\in\asete$ such that $\effect = \eventset|\proc$ for some $\proc\in\procset$, it holds that
\begin{equation}
\label{eq:trivial response function}
    \orep{\effect}(\lambda) = 1;
\end{equation}
\item the ontic response function is additive with respect to coarse-grainings of mutually incompatible observations: for all $\effect_0\in\asete$, $\{\effect_i\in\asete\}_{i=1}^n$ such that there exist $\proc \in\procset$, $\{\event_i\in\sigmaalg\}_{i=1}^n$ satisfying
\begin{subequations}
\label{eq:gen add}
\begin{align}
    \effect_0 &= \cup_{i=1}^n \event_i | \proc, \\
    \forall i \in \{1,\dots,n\}\st \effect_i &= \event_i | \proc, \\
    \forall i \neq j \in \{1,\dots,n\}\st \emptyset &= \event_i \cap \event_j,
\end{align}
\end{subequations}
it holds that
\begin{equation}
\label{eq:ontic response function additivity}
\orep{\effect_0}(\hv) = \textstyle\sum_{i=1}^n \orep{\effect_i}(\hv);
\end{equation}
\item the ontological model implements classical mixtures as ontic mixtures, i.e., $\ostate{}$ is convex-linear on $\convhull\asets$ and $\orep{}$ is convex-linear on $\convhull\asete$, which explicitly reads for $\prepdens = \sum_{\prep\in\asets}\prepdensarg{\prep}\prep$ and $\effectdens = \sum_{\effect\in\asete}\effectdensarg{\effect}\effect$:
\begin{subequations}
\label{eq:convex-linear}
\begin{align}
    \ostate{\prepdens}(\hv) &= \textstyle\sum_{\prep\in\asets} \prepdensarg\prep \ostate{\prep}(\hv), \\
    \orep{\effectdens}(\hv) &= \textstyle\sum_{\effect\in\asete} \effectdensarg\effect \orep{\effect}(\hv);
\end{align}
\end{subequations}
\item the ontological model reproduces the expected statistics:
\begin{equation}
\label{eq:ontic statistics}
    \probarg{\effectdens}{\prepdens} = \textstyle\sum_{\hv\in\ospace} \orep{\effectdens}(\hv)\ostate{\prepdens}(\hv).
\end{equation}
\end{myitem}
\end{definition}
In principle, one could relax this definition to allow for sub-normalized ontic state distributions, i.e., allow for $\sum_{\hv\in\ospace}\ostate{\prepdens}(\hv) \leq 1$, as is discussed in, e.g., \citeref{selby_accessible_2021}, to capture that a preparation procedure can fail to produce the desired result.
Conceptually, we do not feel compelled to embrace this relaxation, since ``no system produced'' ought to correspond to an ontic state of reality where there is ``truly nothing'' (for instance, the ontological model of \citeref{catani_why_2021} includes an ontic state that correspond to a null occupation number of a photonic mode in the arm of an interferometer) but this is of course a possible generalization of our work.

Importantly, there always exists an ontological model for any prepare-and-measure scenario (it is easy to explicitly construct one where the ontic space $\ospace$ is identified with the set of preparations $\asets$ --- this is for instance demonstrated in the proof of \cref{lem:perfect state distinguishability}).
In order for the existence of an ontological model to be non-trivial and potentially useful as a classicality criterion, it is thus necessary to add further constraints.

\paragraph{Relative noncontextuality.}

A relative noncontextual ontological model for a prepare-and-measure scenario $\apm$ with respect to a reference $\rpm$  is an ontological model for $\apm$ that satisfies additional constraints. These are that indistinguishability relations induced by the reference $\rpm$ must be reflected exactly in the ontological model, in the sense that two preparation densities (effect densities) that yield the same statistics as far as $\rsete$ ($\rsets$) is concerned should be represented by the same ontic state distribution (ontic response function) in the ontological model. Formally, using the notation of \cref{sec:introducing indistinguishability} for indistinguishability relations:
\begin{definition}[Relative noncontextual ontological model]
\label{def:classical model}
Let $\apm$ be a prepare-and-measure scenario (\cref{def:pm scenario}) and $\rpm$ be a reference (\cref{def:ref}).
A relative noncontextual ontological model for $\apm$ with respect to the reference $\rpm$ is an ontological model
\begin{equation}
\omodel
\end{equation}
for $\apm$ such that, for all $\prepdens,\prepdens' \in \convhull\asets$, for all $\effectdens,\effectdens' \in \convhull\asete$, the following noncontextuality relations hold:
\begin{subequations}
\label{eq:relative noncontextuality}
\begin{align}
\prepdens \indist{\rsete} \prepdens' &\limplies \ostate{\prepdens} = \ostate{\prepdens'}, \label{eq:preparation relative noncontextuality} \\
\effectdens \indist{\rsets} \effectdens' &\limplies \orep{\effectdens} = \orep{\effectdens'}. \label{eq:effect relative noncontextuality}
\end{align}
\end{subequations}
\end{definition}
This definition allows us to say that a prepare-and-measure scenario is noncontextual relative to a certain reference if it admits a relative noncontextual ontological model with respect to this reference.
Thus, relative noncontextuality is not an absolute property of a prepare-and-measure scenario, but rather a joint property of a prepare-and-measure scenario and a reference.
We will return to the relation between relative noncontextuality and conventional notions of noncontextuality in \cref{sec:choices of references}.

\subsection{Varying the reference}
\label{sec:varying the reference}
 
We now formulate basic results that clarify the consequences of choosing different reference procedures $\rpm$. These are proven in \cref{app:varying the reference}.
First off, in order to allow for the possibility of a relative noncontextual ontological model, it is necessary for the reference to be a faithful reference, i.e., to be a reference that does not coarse-grain the indistinguishability relations induced by the prepare-and-measure scenario itself. 
\begin{restatable}{lemma}{LemNeedFaithful}
\label{lem:need faithful}
Suppose that a prepare-and-measure scenario $\apm$ admits a relative noncontextual ontological model with respect to a reference $\rpm$. Then, the reference $\rpm$ is faithful with respect to $\apm$ in the sense of \cref{def:faithful reference}.
\end{restatable}

Now, with respect to the preorder introduced in \cref{def:ref preorder}, we have that if\linebreak$\rpm \refordereq \rpmm$, then $\rpmm$ is a reference that allows to distinguish the operational preparations and effects of a prepare-and-measure scenario at least as well as $\rpm$ does, and as a result, constructing a relative noncontextual ontological model for a prepare-and-measure scenario with $\rpmm$ is at least as easy as it with $\rpm$.
\begin{restatable}{lemma}{LemInclusionRelations}
\label{lem:inclusion relations}
Consider a prepare-and-measure $\apm$ and two choices of references $\rpm$ and $\rpmm$. Then, if 
\begin{equation}
\label{eq:inclusion condition}
\rpm \refordereq \rpmm,
\end{equation}
then the existence of a relative noncontextual ontological model for $\apm$ with respect to the reference $\rpm$ implies the existence of a relative noncontextual ontological model for $\apm$ with respect to the reference $\rpmm$.
\end{restatable}

An extreme case happens if the reference $\rpm$ is ``so good'' that there does not exist any non-trivial indistinguishability relation. In that case, constructing a relative noncontextual ontological model becomes trivial. This is for instance the case if, for all $\prep\in\asets$, there exists $\effect_\prep\in\rsete$ such that
\begin{equation}
    \forall \prep,\prep' \in \asets \st \probarg{\effect_\prep}{\prep'} = \left\{\begin{aligned}
    1 \textup{ if } \prep = \prep', \\
    0 \textup{ else.}
    \end{aligned}
    \right.
\end{equation}
Visually, this corresponds to the existence of a subset of rows of the data table inducing a submatrix that looks like an identity block --- see for instance \cref{fig:example perfect distinguishability}.
\begin{figure}[ht]
    \centering
    $$
    \drawdatatable{
    \foreach \j in {1,2,3} {
        \node at (0,\j) {$\prep_\j$};
        \node at (\j+2,0) {$\effect_{\prep_\j}$};
    }
    \foreach \j in {2}
        \node at (\j,0) {$\effect_n$};
    \node at (1,0) {$\vdots$};
    \foreach \i/\j/\p in 
    {1/1/\vdots,1/2/\vdots,1/3/\vdots,2/1/0,2/2/\frac13,2/3/1}
        \node at (\i,\j) {$\p$};
    \foreach \i in {1,2,3} {
    \foreach \j in {1,2,3} {
        \ifnum\i=\j
            \node at (\i+2,\j) {$1$};
        \else
            \node at (\i+2,\j) {$0$};
        \fi
    }}
    \drawgrid{5}{3}
    %
\shaderegion{red}{(1,1)}{(5,3)}
\draw[bracestyle,red] (5.25,-1.5) -- node[left,black] {$\rsete$} (0.75,-1.5);
\draw[bracestyle,red] (-1.5,0.75) -- node[above,black] {$\rsets$} (-1.5,3.25);
\shaderegion{blue}{(1.25,1.25)}{(1.75,2.75)}
\draw[bracestyle,blue] (2.25,-0.5) -- node[left,black] {$\asete$} (0.75,-0.5);
\draw[bracestyle,blue] (-0.5,0.75) -- node[above,black] {$\asets$} (-0.5,3.25);
\draw[line width=1pt,black] (2.6,0.8) rectangle (5.4,3.2);
\node (text) at (4,6.5) {perfect distinguishability};
\draw[line width=1pt,black,arrows={->}] (4,3.2) -- (text);
    }
    $$
    \caption{An example prepare-and-measure scenario and reference on a data table such that the assumptions of \cref{lem:perfect state distinguishability} are met, i.e., the reference perfectly distinguishes all the preparations of the prepare-and-measure scenario.}
    \label{fig:example perfect distinguishability}
\end{figure}
In that case, we have that if some $\prepdens,\prepdens'\in\convhull\asets$ satisfy $\prepdens \indist{\rsete} \prepdens'$, then in particular, for all $\prep\in\asets$, it holds that
\begin{equation}
    \probarg{\effect_\prep}{\prepdens} = \probarg{\effect_\prep}{\prepdens'},
\end{equation}
which is equivalent to: for all $\prep\in\asets$,
\begin{equation}
    \prepdensarg{\prep} = \prob_{\prepdens'}(\prep),
\end{equation}
or in other words, $\prepdens = \prepdens'$.
In general, we can write the following lemma:
\begin{restatable}{lemma}{LemPerfectStateDistinguishability}
\label{lem:perfect state distinguishability}
Let $\apm$ be a prepare-and-measure scenario, and let $\rpm$ be a faithful reference that is such that for all $\prepdens,\prepdens'\in\convhull\asets$,
\begin{equation}
\label{eq:perfect preparation distinguishability condition}
    \prepdens \indist{\rsete} \prepdens' \limplies \prepdens = \prepdens'.
\end{equation}
Then, $\apm$ admits a relative noncontextual ontological model with respect to the reference $\rpm$.
\end{restatable}
The proof is also given in \cref{app:varying the reference}.

\begin{figure}[ht]
    \include{typical_graph}
    \caption{A typical relative noncontextuality graph. Given a fixed prepare-and-measure scenario $\apm$, each vertex of the graph is associated with a reference $\rpm$, and the vertex style represents the answer to the question ``does $\apm$ admit a relative noncontextual ontological model with respect to the reference $\rpm$?'': a ``yes'' is depicted as 
    $\centertikz{\protect\node[clnode] {};}$, 
    and a ``no'' is depicted as $\centertikz{\protect\node[nclnode] {};}$. A directed edge from a reference $\rpm$ to another reference $\rpmm$ indicates the preorder relation $\rpm\refordereq\rpmm$ --- for readability, we omit the directed edge from a vertex to itself, and we also omit the directed edges that can be obtained from the transitivity of the preorder relation.
    Because of \cref{lem:need faithful}, we omit any choice of reference that is not faithful, which implies that for any drawn reference $\rpm$, the directed edge from $\apm$ to $\rpm$ exists.
    \Cref{lem:inclusion relations} forbids the existence of any $\centertikz{\protect\node[clnode] (a) at (0,0) {}; \protect\node[nclnode] (b) at (1,0) {};\protect\draw[preorder] (a) -- (b);}$ edge. 
    If we include reference choices that satisfy the assumptions of \cref{lem:perfect state distinguishability}, then the corresponding vertices (high up in the tree) are necessarily of $\centertikz{\protect\node[clnode] {};}$ type.
    Broken arrows $\centertikz{\protect\draw[dashed,black!90,line width=0.7pt, arrows={-Latex}] (0,0) -- (15pt,0);}$ indicate additional information about a particular vertex.
    }
    \label{fig:typical graph}
\end{figure}

\subsection{Visualizing changes of reference: the relative noncontextuality graph}
\label{sec:relative nc graph}

We now introduce a convenient way of visualizing the effect of choosing different references. Given a prepare-and-measure scenario $\apm$, we introduce a directed graph $G = (V,E)$ whose vertex set is a (finite) set of references, i.e., $V = \left\{\pmargs{\rsets^{(i)}}{\rsete^{(i)}}\right\}_i$.
For any two vertices $v, v' \in V$, the directed edge $(v,v')$ belongs to the edge set $E$ if and only if $v \refordereq v'$, i.e., if the reference $v'$ distinguishes the procedures of $\apm$ at least as well as the reference $v$ does.
Strictly speaking, one should rather associate the vertices in $V$ with equivalence classes of references that induce the same indistinguishability relations on the prepare-and-measure scenario: this is a pragmatic choice since the reference only enters the definition of a relative noncontextual ontological model through the indistinguishability relations that the reference induces, and it furthermore guarantees that the graph $G$ is acyclic.\footnote{For any set $V$ with preorder relation $\preceq$, the directed graph $G = (V,E)$ where $(v,v') \in E \lequiv v\preceq v'$ may in general be cyclic. 
However, let $\sim$ be the equivalence relation induced from $\preceq$ as $v \sim v'$ if and only if both $v\preceq v'$ and $v'\preceq v$, and let $\tilde V$ be the set of equivalence classes of $V$ with respect to $\sim$, where for any $v\in V$ the associated equivalence class is $[v] \in \tilde V$.
Then, consider the quotiented graph $\tilde G = (\tilde V,\tilde E)$ where $([v],[v']) \in \tilde E$ if and only if (i) $v \preceq v'$ and (ii) $[v] \neq [v']$: this graph is acyclic by construction. Indeed, by transitivity of the preorder relation, a cycle would imply that there exists a directed edge from $[v]$ to itself, contradicting the assumption (ii).}
%
%
We then equip the graph with a Boolean function $f : V \to \{\centertikz{\node[clnode] {};},\centertikz{\node[nclnode] {};}\}$ such that $f(v) = \centertikz{\node[clnode] {};}$  if $\apm$ admits a relative noncontextual ontological model with respect to the reference $v$ and $f(v) = \centertikz{\node[nclnode] {};}$ else.
A typical such graph $G = (V,E)$ and Boolean function $f$ is shown in \cref{fig:typical graph}, where the main features that this graph inherits from \cref{lem:need faithful,lem:inclusion relations,lem:perfect state distinguishability} are presented. Importantly, the graph shown in \cref{fig:typical graph} is obtained from a prepare-and-measure scenario $\apm$ that does not admit a relative noncontextual ontological model with respect to the reference $\rpm = \apm$, for otherwise, \cref{lem:inclusion relations} proves that the graph would only ever feature $\centertikz{\node[clnode] {};}$ nodes (as long as all references of the graph are faithful ones).
We include in our graph some choices of references that satisfy the assumptions of \cref{lem:perfect state distinguishability}: there exists a relative noncontextual ontological model for $\apm$ with respect to these references, so that the corresponding vertices in the relative noncontextuality graph are $\centertikz{\node[clnode] {};}$ vertices.
Such vertices are typically relatively high up the tree, since they correspond to ``good'' references with respect to distinguishing the preparations. They are not necessarily ``terminal vertices'' since other choices of references could improve the distinguishability of the effects.
In \cref{fig:typical graph}, several $\centertikz{\node[nclnode] (a) at (0,0) {}; \node[clnode] (b) at (1,0) {};\draw[preorder] (a) -- (b);}$ edges are present: each of these corresponds to a transition from contextuality to noncontextuality, or, conceptually, to a transition from non-classicality to classicality.
In general, we do not expect there to be a single such edge: multiple inequivalent transitions from non-classicality to classicality are logically possible for a given prepare-and-measure scenario.

\section{Choices of references and notions of systems}
\label{sec:choices of references}

\subsection{Preliminaries}

Given a prepare-and-measure scenario $\apm$, we have so far only described the question of whether or not $\apm$ admits a relative noncontextual ontological model with respect to some choice of reference $\rpm$.
Importantly, different choices of references are typically associated to different answers to this question --- the generic case is depicted in \cref{fig:typical graph}.
If $\rpm$ is sufficiently large, there will be no statistical equivalences at all: the smallest operational distinction between two operational procedures can always be witnessed.
As \citeref{pusey_robust_2018} puts it, 
\say{
it is always logically possible that there exists a measurement that simply reads out a complete description of the preparation. In that case, no two distinct preparations would be equivalent, and so any model would be trivially noncontextual}{4} --- see also \cref{lem:perfect state distinguishability}.
Operationally, such a measurement could be implemented by looking at suitable video recordings of what the agent did in the lab.
To obtain a meaningful notion of classicality from relative noncontextuality, one thus has to motivate a restriction of the scope of procedures going into the reference $\rpm$, but moreover, it can be desirable to fix the reference once and for all when looking at a given prepare-and-measure scenario.
The rest of this section is dedicated to reviewing different proposals for this purpose.

\paragraph{System-based approaches to noncontextuality}

Most of the proposals of this section will fall into what we would like to call system-based approaches to noncontextuality. 
The overarching idea of such approaches is to use some definition of system that relates to the prepare-and-measure scenario and that furthermore suggests a candidate reference $\rpm$.
Depending on the relations that one posits between the prepare-and-measure scenario, the reference and the system, this yields different approaches to noncontextuality.
It is important to already highlight that a physical theory, such as quantum theory, does not provide an absolute definition of system.
In non-relativistic quantum mechanics, the typical situation is that each experiment will be described through operators acting on a Hilbert space. The latter is thus eligible as defining indirectly the system. 
However, systems need not be in one-to-one correspondence with Hilbert spaces: it is, to some extent, up to the practitioner of quantum mechanics to place the cut between the system and the environment, between the relevant and irrelevant procedures, etc.
Even under the assumption that each system should correspond to some Hilbert space, this does not solve the bulk of the issue: any prepare-and-measure scenario can be physically embedded into several inequivalent Hilbert space representations.

\subsection{The limits of in-principle indistinguishability}
\label{sec:limits of in-principle indistinguishability}

Before turning to the actual proposals, it is useful to fix, to some extent, the terminology that we shall use, as well as some methodological elements.
In \citeref{spekkens_ontological_2019}, noncontextuality is motivated based on a methodological principle traced back to Leibniz.
Here, we would like to discuss one specific point that is relevant for our purposes: namely, that \say{
the Leibnizian methodological principle does not appeal to a parochial kind of empirical indiscernibility, judged relative to the particular in-born capabilities of humans or their particular technological capabilities at a given historical moment, but rather to the in-principle variety of empirical indiscernibility}{8, \cite{spekkens_ontological_2019}}.
We argue that it is conceptually and methodologically advantageous for our purposes to consistently replace the phrase ``in-principle indiscernibility'' with something along the lines of ``indiscernibility with respect to all procedures pertaining to the system''.

\begin{figure}[ht]
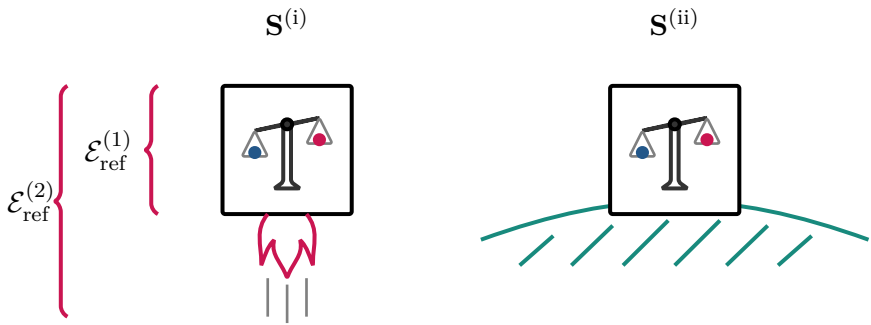

\include{elevator}
\caption{Inspired from \citeref{spekkens_ontological_2019}. Consider the preparations $\prep^\text{(i)} =$ ``perform an experiment in a lab accelerated to $+g$ by means of a rocket'' and $\prep^\text{(ii)} =$ ``perform the same experiment in a lab on the surface of Earth whose gravitational acceleration is $-g$''. The two preparations are indistinguishable with respect to the measurements in the set $\rsete^{(1)}$ that only look at the inside of the lab, but are distinguishable with respect to the measurements in the set $\rsete^{(2)}$ that look outside of the lab.}
\label{fig:gr dist}
\end{figure}

\paragraph{Distinguishability in physical theories.}

For instance, consider the case of a closed lab that is either (i) accelerated vertically by means of a rocket attached to the bottom of the lab or (ii) at rests on the surface of Earth, such that the acceleration felt inside the lab is identical in the two cases, as discussed in Section I.B of \citeref{spekkens_ontological_2019}, and which we illustrate in \cref{fig:gr dist}. Clearly, the type of indiscernibility relation that relates the cases (i) and (ii) rests on the assumption that the agent in the lab does not look outside the lab. Thus, rather than thinking of ``in-principle indiscernibility between (i) and (ii)'', we would rather think of ``indiscernibility between (i) and (ii) with respect to arbitrary procedures pertaining to the spatial extent of the inside of the lab''. 
This is in fact consistent with typical statements of the strong equivalence principle \cite{haugan_principles_2001} that entail restrictions on the type of experiments (typically referred to as ``local experiments'') that should be considered when formulating equivalences between scenarios such as (i) and (ii).
These considerations extend more generally to other theories such as quantum theory, which does not contain fundamental indistinguishability relations, but rather indistinguishability relations with respect to some system (or, more generally, with respect to a restriction of the scope of experiments that are allowed to be used to verify the indistinguishability relation).
Thus, we would like to negate, or at least question, the idea that, in the context of indistinguishability, \say{what is possible in principle is determined
by the physical theory that one is assessing}{5, \cite{catani_reply_2022}}. For instance, in the example of \cref{fig:gr dist}, general relativity does not prevent one from analyzing light coming into the lab from the outside and distinguishing the scenarios (i) and (ii).
Similarly, quantum theory does not prevent one from distinguishing between the preparation described as ``prepare a uniform mixture of the two orthogonal states $\ket0$ and $\ket1$'' and ``prepare a uniform mixture of the two orthogonal states $\ket\pm = (\ket0 \pm \ket1)/\sqrt 2$'': indeed, to distinguish the two, it is always possible to look at video recordings of what happened in the lab, or to investigate side effects over the environment of the different preparations corresponding to $\ket0$, $\ket1$ and $\ket\pm$ --- e.g., it could be that the preparations in fact really prepare states of the form $\ket 0 \otimes \ket{\psi_0}_E$, $\ket 1 \otimes \ket{\psi_1}_E$, $\ket\pm \otimes \ket{\psi_\pm}_E$ where $\ket{\psi_0}_E$, $\ket{\psi_1}_E$ and $\ket{\psi_\pm}_E$ are four orthogonal states of the environment (in fact, the case of the video recording falls into this category upon identifying the possible contents of the video recording with these four perfectly distinguishable states).

\paragraph{Applications to noncontextuality.}
Of course, one can define in-principle indistinguishability as a contextual statement where the context provides the relevant notion of system, but this does not fit our goal of interrogating the interplay between different notions of systems and noncontextuality.
For instance, we do not embrace without additional considerations the statement that
\say{
generalized noncontextuality only implies constraints on the ontological representations of operational processes which are genuinely operationally equivalent --- that is, those that give exactly the same predictions for any operational scenario in which they appear. For example, two preparation procedures are only operationally equivalent if they give the same statistics for the outcomes of all possible measurements
}{20, \cite{selby_open-source_2022}}.
Such a statement only makes sense with respect to an underlying notion of system, 
and we wish to make this assumption explicit in the absence of a consensus on the status of this assumption.
Furthermore, even if the community comes to a consensus on the appropriate notion of system that one should use in this context, there remains the question of whether one can provide a definitive methodology for identifying the specific system that is relevant for a given prepare-and-measure scenario.

\subsection{Operationally noncontextual ontological models}
\label{sec:operational noncontextuality}

A simple approach is to choose the prepare-and-measure scenario $\apm$ itself as the reference. This is, presumably, the approach that requires the least amount of additional considerations and inputs from other notions, such as notions of system. In that sense, one can think of this approach as being an approach to noncontextuality that does not involve systems.
Alternatively, one can think of this approach as using a notion of system in which any prepare-and-measure scenario can be thought of as defining a system that does not extend beyond the prepare-and-measure scenario. 
This yields the following definition:

\begin{definition}[Operational noncontextuality]
\label{def:operational noncontextuality}
A prepare-and-measure scenario\linebreak$\apm$ is said to be operationally noncontextual if it admits a relative noncontextual ontological model with respect to the reference being $\apm$ itself, otherwise, $\apm$ is said to be operationally contextual.
\end{definition}

In fact, this choice turns out to be equivalent to the reduced space approach of \citeref{gitton_solvable_2022}.\footnote{Strictly speaking, operational noncontextuality as defined in \cref{def:operational noncontextuality} is a slight generalization of \citeref{gitton_solvable_2022}: the difference lies in whether or not one imposes that effects corresponding to a complete set of outcomes are all included in the prepare-and-measure scenario.}
There, a so-called operationally noncontextual ontological model was formulated for a quantum (or, more generally, GPT) model of a prepare-and-measure scenario.
Establishing the equivalence thus relies on formulating a sensible definition of a quantum representation of an operational prepare-and-measure scenario as defined in the present \cref{def:pm scenario}.
Consider the following informal equivalence statement, 
which is formalized and proven in \cref{sec:proof reduced space equivalence} --- it includes in particular a proof that any prepare-and-measure scenario admits a quantum model. 
\begin{theorem}[Informal]
\label{th:reduced space equivalence}
Consider a prepare-and-measure scenario $\apm$, and consider any quantum model of $\apm$.
$\apm$ is operationally noncontextual if and only the quantum model of $\apm$ admits an operationally noncontextual ontological model in the sense of \citeref{gitton_solvable_2022}.
\end{theorem}

Other approaches to noncontextuality can occasionally agree with operational noncontextuality: this is for instance the case of \citerefs{ferrie_frame_2008,spekkens_negativity_2008}, where it is shown that a prepare-and-measure scenario corresponding to all quantum states and measurements on some Hilbert space does not admit a noncontextual ontological model with respect to the indistinguishability relations that can be read off from the Hilbert space representation of the prepare-and-measure scenario --- the proof is explicitly reproduced is section 3.3.2 of \citeref{gitton_solvable_2022}.

\paragraph{Is nature operationally contextual?}

One can say, for instance, that nature is indeed operationally contextual because it admits prepare-and-measure scenarios $\apm$ that are operationally contextual \cite{ferrie_frame_2008,spekkens_negativity_2008}. Looking back at the relative noncontextuality graph of \cref{fig:typical graph}, we see that this is the most easily achieved type of contextuality, since we do not allow the reference to be any better than what the prepare-and-measure scenario itself provides.
In that sense, it is sensible to say that nature is contextual according to this approach, but this claim is the easiest to establish, compared to some other system-based approaches that we describe below.
On the other hand, compared to these other approaches, the operational contextuality conclusions are robust in the sense that the reference is uniquely and unambiguously identified given the prepare-and-measure scenario.

\subsection{Pragmatic-system approaches to noncontextuality}
\label{sec:pragmatic noncontextuality}

In this section, we describe an approach to noncontextuality that is based on a pragmatic notion of system.
A pragmatic system is based on an operational notion of system
in which a system is defined as a set of operational procedures that are thought of as ``probing primarily the system''.
The complement of these operational procedures correspond to those procedures that ``probe the system as well as other systems''.
Obviously, this definition is not particularly formal, but it should be sufficient for our purposes.
This definition is supported for instance by \citeref{mueller_testing_2021}, which suggests to 
\say{
think of physical systems as being defined by an experimental scenario
}{19}.
A pragmatic system is a construction on top of an operational system: namely, a pragmatic system is an operational system together with a justification of its importance. Such a justification typically involves relating the set of procedures that defines the pragmatic system to other notions such as, e.g., the richness of a pragmatic system's phenomenology, the relation of this system to interesting constraints or interesting tasks, etc.

\paragraph{Examples of pragmatic systems.}

For instance, based on the collective interest that we have in photonic experiments, it is pragmatically justified to think of a photon polarization system as corresponding to procedures that \say{[probe] a photon’s polarization in conventional ways (using wave plates and beam splitters)}{2, \cite{mazurek_experimentally_2021}}. 
Furthermore, one can think of discriminating between different pragmatic systems based on statistical evidence: this is for instance a possible reading of some of the considerations of \citeref{mazurek_experimental_2016}, whose system of interest is the polarization of a single photon, and where other choices of systems related to multiple photons are excluded statistically by experimentally demonstrating the rarity of multi-photon events.
In \citeref{spekkens_contextuality_2005}, 
it is stated that \say{An equivalence class of preparation procedures is associated with a density operator $\rho$. This is a positive trace-1 operator over the Hilbert space $\hil$ of the system}{3}: this can be seen as a pragmatic definition of quantum mechanical systems being associated with some Hilbert space.
This is a generic feature of pragmatic systems: a physical theory and the associated phenomenology contribute towards understanding certain systems as potentially more fundamentals than others.
For a different flavor of pragmatism, \citeref{spekkens_preparation_2009} considers a system defined by an operational constraint, namely, a parity-oblivious constraint, that restricts the scope of possible measurements: \say{The parity-oblivious constraint requires that [...] there is no outcome of any measurement for which posterior probabilities for [parity] 0 and [parity] 1 are different}{2}. Such a system inherits its relevance from the relevance of the parity-oblivious constraint.

\paragraph{Pragmatic-system approaches to noncontextuality.}

Following the idea that the reference of indistinguishability to be used in a relative noncontextual ontological model should be associated with an underlying system, a pragmatic system is eligible to be used as such. This gives the following notion of noncontextuality:
\begin{definition}[Pragmatic noncontextuality]
\label{def:pragmatic-system noncontextuality}
A prepare-and-measure scenario $\apm$ is said to be pragmatically noncontextual if, upon identifying a pragmatic system to which $\apm$ belongs, $\apm$ admits a relative noncontextual ontological model with respect to the reference $\rpm$ that contains all relevant procedures from the underlying pragmatic system. Otherwise, $\apm$ is said to be pragmatically contextual.
\end{definition}
In general, the relevance of a statement of pragmatic noncontextuality is inherited from 
the extent to which the pragmatic system is well-chosen. As such, there is an influence from the relevance of a pragmatic system to the relevance of pragmatic noncontextuality claims.
However, if pragmatic noncontextuality is deemed a relevant notion, there may also be a back-action: for instance, if one prefers to have a pragmatic noncontextual ontological model for a prepare-and-measure scenario, then this could suggest to adjust one's pragmatically chosen system to allow for such a pragmatic noncontextual ontological model.
For instance, in \citeref{spekkens_evidence_2007}, a toy theory that can reproduce some quantum phenomena is developed. This toy theory can be seen as a noncontextual ontological model for a subset of procedures on a qubit (see, e.g., appendix D of \citeref{schmid_characterization_2021}).
Furthermore, in \citeref{spekkens_evidence_2007}, it is claimed that the toy theory, being a theory of epistemic states, brings us to a better conceptual understanding of \say{a great number of quantum phenomena that are mysterious from the ontic viewpoint}{2}.
This supports the idea that noncontextual ontological models are useful for our understanding of nature. As a byproduct, this supports the idea that one's notion of pragmatic systems should not be completely agnostic of considerations of pragmatic noncontextuality.

\paragraph{Is nature pragmatically contextual?}

An important property of a pragmatic system is the extent to which it is thought of as being ``fundamental''.
In the literature, pragmatic systems associated to some Hilbert space are frequently considered.
These are typically considered as being rather fundamental. For instance, in \citeref{spekkens_contextuality_2005}, it is said that 
\say{This terminology allows one to use the phrase \emph{“quantum theory is contextual”} as a shorthand for \emph{“quantum theory does not admit a noncontextual ontological model”}}{3}.
This statement may suggest that \citeref{spekkens_contextuality_2005} 
considers that quantum theory is primarily concerned with a particular choice of pragmatic systems, namely, those associated with some Hilbert space $\hil$: recall that there, \say{[a]n equivalence class of preparation procedures is associated with a density operator $\rho$. This is a positive trace-1 operator over the Hilbert space $\hil$ of the system}{3, \cite{spekkens_contextuality_2005}}.
Since there exist pragmatically contextual prepare-and-measure scenarios when the pragmatic system is associated to some Hilbert space \cite{spekkens_contextuality_2005,ferrie_frame_2008,spekkens_negativity_2008}, this shows that nature is indeed pragmatically contextual with respect to this type of pragmatic systems.
However, there remains the subjectivity of one's choice of pragmatic Hilbert space to associate to a prepare-and-measure scenario: as \citeref{mueller_testing_2021} puts it, \say{tomographic completeness is contingent on the way that we define the boundaries of the physical system of interest. For example, doing tomography on the polarization of a photon necessarily ignores all other aspects of its wavefunction; all physical objects are embedded in larger environments that have to be disregarded}{19}.
Furthermore, as we discussed in the previous paragraph, the conceptual usefulness of having access to a noncontextual ontological model for a prepare-and-measure scenario \cite{spekkens_evidence_2007} could potentially be used to argue against this type of pragmatic system as being the most adequate choice.
This is indeed a major distinction between (i) the pragmatic systems corresponding to local experiments \cite{haugan_principles_2001} that appear in the equivalence principle, and (ii) the pragmatic systems corresponding to Hilbert spaces that appear in traditional proofs of contextuality. While the choice of case (i) allowed for (and, in fact, fuelled) the development of general relativity, the choice of case (ii) forbids the use of noncontextual ontological models as alternative descriptions of quantum mechanics.
As such, a proponent of the pragmatic approach to noncontextuality that wishes to infer from traditional proofs of contextuality \cite{spekkens_contextuality_2005,ferrie_frame_2008,spekkens_negativity_2008} that nature itself is contextual should further motivate the reasons for considering other choices of pragmatic system assignments to prepare-and-measure scenarios as being unreasonable.

\subsection{Ontic-system approaches to noncontextuality}
\label{sec:ontic-system noncontextuality}

Let us turn to another system-based approach to noncontextuality.
Recall that in the previous notions of pragmatic systems, systems were defined primarily through a choice of operational procedures that were then justified and related to other conceptual notions: in that sense, the notion of procedures preceded the notion of system.
One can approach the notion of system from a different angle, and argue that a system is a conceptual entity that exists independently of our attempts to probe it: in that sense, a system is seen as more fundamental than potential procedures that probe it. Let us refer to such a system as an ``ontic system''.
The name reflects the fact that the cut between the system and its environment, in this approach, is seen as a true, ontic property of nature.
To each ontic system, we attach a complete set of operations that probe the ontic system: let us refer to this set as the ``ontic set of operations''.
In general, this set contains operations that have not yet been performed, for which we do not have a description yet, and whose derived statistics are unknown.
For instance, let us consider the investigations of \citeref{mazurek_experimentally_2021} with respect to the polarization degree of freedom of a photon. In an ontic-system interpretation of the polarization degree of freedom of a photon, it is sensible to expect that \say{The full scope of possible preparations and measurements for photon polarization might be radically different from what our quantum expectations dictate (incorporating new exotic procedures)}{3, \cite{mazurek_experimentally_2021}}: in other words, in general, the ontic set of operations probing an ontic system can incorporate procedures that are not necessarily known to present day.
Presumably, one's operational characterization of an ontic system is in general a non-empty inner approximation of the ontic set of operations: this relies on the fact that, typically, one trusts that some known operations really probe the ontic system under consideration. To support this claim, one may for instance 
\say{presume that there is a principle of individuation
for different degrees of freedom, which is to say a way to
distinguish what degree of freedom an experiment is probing. For instance, we presume that we can identify certain
experimental operations as preparations and measurements
of photon polarization and not of some other degree of
freedom}{9, \cite{mazurek_experimentally_2021}}.

\paragraph{The example of a particle.} 

To give an explicit example, consider a particle, and in particular, three flavors thereof: a classical particle, described through its position and momentum, a non-relativistic quantum particle, described with its wavefunction, and a hypothetical post-quantum particle. 
In the pragmatic approach to systems, the classical, quantum and post-quantum particle are three valid and distinct pragmatic systems that are defined, primarily, through the associated operations. The reason why these three pragmatic systems are distinct is because they are associated to different operations. These three systems are however not necessarily unrelated, since one can see a classical particle as corresponding to a restricted set of operations on a quantum particle, and one may argue that the quantum particle is a more fundamental pragmatic system than the classical-particle pragmatic system.
On the other hand, in the ontic-system approach, a possible ontic system could be ``a particle''. The operations associated to a classical particle would be a naive attempt to try to probe a particle, but more fundamentally, the operations associated to a quantum particle probe more properties of the particle.
Whether or not the quantum operations associated to a quantum particle completely capture the essence of the particle depends on the existence of a post-quantum theory that would come with a notion of particle and new operations on it.

\paragraph{Ontic-system noncontextuality.}

We now discuss what an ontic-system approach to noncontextuality could look like. Consider this quote from \citeref{pusey_contextuality_2019}:
\say{The operational approach to contextuality due to Spekkens requires finding operationally equivalent preparation procedures \cite{spekkens_contextuality_2005}. Previously these have been obtained by demanding indistinguishability under a set of measurements taken to be tomographically complete. In the language of generalised probability theories, this requires the ability to explore all the dimensions of the system’s state space. However, if the true tomographically complete set is larger than the set assumed, the extra measurements could break the operational equivalences and hence eliminate the putative contextuality. Such extra dimensions could arise in post-quantum theories, but even if quantum theory is exact there can be unexpected degrees of freedoms due to imperfections in an experiment}{1}.
In the language developed in this section, one may understand this quote as the following approach to noncontextuality.
To each prepare-and-measure scenario, one should first have a prescription of which ontic system the prepare-and-measure scenario is probing.
Because the system associated to the prepare-and-measure scenario is an ontic system, it is a conceptual entity whose ontic set of operations --- its \saynp{true tomographically complete set} --- is, in general, unknown.
Turning to noncontextuality, the reference for indistinguishability that one uses in noncontextuality should be the ontic set of operations of the ontic system --- in particular, equivalences between preparations should be based on \saynp{indistinguishability under a set of measurements taken to be tomographically complete}.
We can summarize these considerations in the following tentative definition:
\begin{definition}[Ontic-system noncontextuality]
\label{def:ontic-system noncontextuality}
A prepare-and-measure scenario\linebreak$\apm$ is ontic-system noncontextual if $\apm$ can be seen as part of an ontic system and if $\apm$ admits a relative noncontextual ontological model with respect to the reference $\rpm$ made out of the ontic set of operations corresponding to the relevant ontic system. Otherwise, $\apm$ is said to be ontic-system contextual.
\end{definition}

\paragraph{Is nature contextual in the ontic-system approach to noncontextuality?}

The ontic-system approach to noncontextuality is advantageous with respect to investigations of nature's contextuality in the following sense: because the ontic systems are understood as being fundamental entities independent of our technological capabilities, if there exists at least one ontic system that is contextual, then it is sensible to say that nature is contextual.
However, to establish a claim of ontic-system contextuality, one needs to assume that one's reference of indistinguishability captures all of the ontic set of operations associated to the ontic system.
Even under the strong assumption that all of our theories of nature are final and that there will be no significant update thereof, it seems hard to falsify the hypothesis that a certain set of operations really corresponds to the ontic set of operations associated to the ontic system: from an operational perspective, the separation between an ontic system and ``the rest'' is a subjective cut that is elevated to an ontic property of nature.
In particular, this cut is not forced upon us by a physical theory, as discussed in \cref{sec:limits of in-principle indistinguishability}.
It would also be interesting to discuss whether or not the notion of ontic system is reconcilable with the intuition that the ontic reality takes places on a continuous spacetime background.

\begin{figure}[ht]
    \include{micro_cut}
    \caption{Consider a procedure $\prep$ that Alice implements by interacting with a classical device which in turns interacts with a microscopic system. Interactions propagating in space are represented by $\centertikz{\protect\draw[line width=1pt,draw=black!80,arrows={->}] (0,0) -- (1,0);}$ arrows, and the arrows between the classical device and the macroscopic system should be understood as possible interactions (one may not know exactly how the information propagates to the microscopic system). Suppose that one attaches a spatial region $R$ to the procedure $\prep$: we represent $R$ as the shaded red region. This region necessarily looks like a cut in the measurement chain, i.e., a chosen subset of the spatial regions that must be influenced by the procedure $\prep$. A reasonable such cut can presumably be placed on the right of the classical device, but this does not guarantee the uniqueness of the cut.}
    \label{fig:measurement cut}
\end{figure}

\subsection{Locality considerations}
\label{sec:local noncontextuality}

Can locality considerations help in this context?
It is relatively clear that thinking about the local structure of an experiment can help to enrich one's notion of pragmatic-system noncontextuality, and can support one's notion of ontic-system noncontextuality.
Additionally, there exist relations between Bell non-locality and noncontextuality: for instance, as discussed in section V of \citeref{yadavalli_contextuality_2021}, Bell non-locality implies preparation contextuality.
It appears unclear, however, that one can entirely dismiss the above discussion of pragmatic and ontic systems in favor of some notion of locality.
Suppose for instance that one uses the following methodology of ``local noncontextuality'': given a prepare-and-measure scenario, (i) identify the region of spacetime over which the prepare-and-measure scenario takes place on, (ii) choose the reference made out of all the procedures pertaining to this region of spacetime, and (iii) investigate the existence of a relative noncontextual ontological model with respect to this reference.
An issue in this methodology is step (i): assigning a precise spacetime extent to a given procedure is rather subjective. Given any tentative spacetime region $R$ that one assigns to a given procedure, as is typically done in algebraic quantum field theory \cite{fewster_algebraic_2019}, there is necessarily information flowing from the agent to $R$, such that $R$ in fact looks like a subjective (although, of course, potentially pragmatic) cut in a measurement chain: we represent this situation in \cref{fig:measurement cut}.
This cut is possibly specified at a microscopic scale.
This is rather distinct from the locality considerations of, e.g., Bell scenarios \cite{bell_einstein_1964,clauser_proposed_1969}: there, the locality assumptions have to do with the spacelike separation of two macroscopic regions assigned to Alice and Bob.
If one tackles steps (i)-(ii) pragmatically, i.e., using as a reference a common, interesting choice of procedures that are typically taken to probe the spatial extent of the prepare-and-measure scenario, then one falls back on an instance of pragmatic-system noncontextuality.
Assuming that one obtains a verdict of contextuality, this opens up the door to discussions about whether there is a useful update of one's notion of (microscopic) locality that could be used to allow for a noncontextual explanation of the prepare-and-measure scenario.

\newpage
\section{Conclusion}
\label{sec:conclusion}

In \cref{sec:pm scenarios}, we introduced an operational formulation of prepare-and-measure scenarios including a reference for indistinguishability, i.e., including an explicit choice of operational procedures that one uses to talk about indistinguishability of the operational procedures belonging to the prepare-and-measure scenario.
This allowed us, in \cref{sec:relative noncontextuality}, to formulate noncontextual ontological models in a way where the indistinguishability relations that one uses are explicitly tied to a choice of reference for indistinguishability. We named these models relative noncontextual ontological models --- see \cref{def:classical model}.
Two choices of references are occasionally comparable if one induces a fine-graining or coarse-graining of the indistinguishability relations that the other induces.
Relative noncontextual ontological models are, in some sense, monotone under such fine-grainings or coarse-grainings: a better reference that can distinguish more procedures of the prepare-and-measure scenario makes it more likely for there to exist a relative noncontextual ontological model.
Conceptually, the generic situation is depicted in \cref{fig:typical graph}, which shows whether or not a given prepare-and-measure scenario admits a relative noncontextual ontological model with respect to different choices of references in a way compatible with the fine-graining relations that hold between different choices of references.
Ultimately, an arbitrarily good reference will allow for the existence of a relative noncontextual ontological model --- see \cref{lem:perfect state distinguishability}.
In order for relative noncontextual ontological models to be a useful conception of classicality for a prepare-and-measure scenario, it is thus necessary to restrict the scope of the reference to some extent.

\paragraph{Fixing the reference.}

We offered a few approaches towards answering the question of how to meaningfully fix the reference used in relative noncontextual ontological models in \cref{sec:choices of references}. We first used \cref{sec:limits of in-principle indistinguishability} to justify the idea that in-principle indistinguishability is, in this context, not descriptive enough.
Then, the simplest idea described in \cref{sec:operational noncontextuality} is to take the procedures of the prepare-and-measure scenario itself to act as a reference for indistinguishability.
The remaining two approaches are related to notions of systems. We distinguish two notions of systems: pragmatic systems are any set of operational procedures endowed with some relevance that is established pragmatically, while ontic systems are conceptual entities that pre-exist our attempts to probe them, and whose ontic set of operations (the set of operations that would probe the entire extent of the ontic system) is in general unknown.
The reference of indistinguishability can then be chosen according to the following tentative methodologies, respectively: given a prepare-and-measure scenario, choose an interesting pragmatic system to which the prepare-and-measure scenario belongs and use the procedures defining the pragmatic system as the reference for indistinguishability in relative noncontextual ontological models, or identify an ontic system to which the prepare-and-measure scenario belongs and use the best guess one has for the ontic set of operations of the ontic system to generate the reference of indistinguishability.

\paragraph{Outlook.}

This work did not discuss transformations, although these are also eligible to be described in a noncontextual ontological model \cite{spekkens_contextuality_2005}: it would be interesting to reflect further on possible insights drawn from the case of transformation noncontextuality.
Practically speaking, it could be of interest to adapt existing algorithms for testing noncontextuality \cite{schmid_all_2018,selby_open-source_2022} to the general case of relative noncontextuality.
With the exception of operationally noncontextual ontological models (\cref{sec:operational noncontextuality}), the approaches to fix the reference of indistinguishability that we described have not been fully developed and require further discussion and motivation.
Some concrete open questions are the following: in the approach of pragmatic noncontextuality, are there prepare-and-measure scenarios traditionally deemed contextual that could be reassessed as being noncontextual with respect to another choice of pragmatic system acting as the reference of indistinguishability? 
More generally, in both the pragmatic and ontic-system approaches to noncontextuality, to what extent is contextuality an artifact of our choice of partition of the world into subsystems, and how well can one motivate any such partition?
In light of the discussion of \cref{sec:local noncontextuality}, discussing these questions may provide an interesting prism for comparing different notions of microscopic locality.
With respect to the existing literature, we hope that our work will help to open up new avenues for clarifying the foundations of noncontextuality. 
This could for instance involve a conceptual and methodological enrichment of some of the approaches for fixing the reference that we developed in \cref{sec:choices of references}, but it could well be that other approaches defying our attempted categorization are of greater interest. 
Of course, one can always question the virtue of wanting to fix a reference in the first place. 
It is possible that there will never be a satisfactory methodology for doing so and that it is more meaningful to examine pragmatically chosen noncontextuality graphs (see \cref{fig:typical graph}) for a given experiment instead. 

\section*{Acknowledgments}

We thank Nuriya Nurgalieva and Martin Sandfuchs for useful feedback. This work was supported as a part of NCCR QSIT, a National Centre of Competence (or Excellence) in Research, funded by the Swiss National Science Foundation (grant number 51NF40-185902).
M.W.\ acknowledges support from the Swiss National Science Foundation (SNSF) via an AMBIZIONE Fellowship (PZ00P2\_179914).

\bibliographystyle{mybibstyle}
\bibliography{origins}

\onecolumn\newpage
\appendix

\section{Varying the reference in relative noncontextuality}
\label{app:varying the reference}

In this section, we restate and prove the results of \cref{sec:varying the reference}.

\LemNeedFaithful*
\begin{proof}
\Cref{eq:ontic statistics,eq:relative noncontextuality} together imply that, for all $\prepdens,\prepdens'\in\convhull\asets$, for all $\effectdens,\effectdens'\in\convhull\asete$,
\begin{equation}
    \prepdens \indist{\rsete} \prepdens' \textup{ and } \effectdens \indist\rsets \effectdens' \limplies 
      \probarg{\effectdens}{\prepdens} = \probarg{\effectdens'}{\prepdens'}.
\end{equation}
In particular, 
this implies that
\begin{equation}
    \prepdens \indist{\rsete} \prepdens' \limplies \forall \effectdens\in\convhull\asete \st \probarg{\effectdens}{\prepdens} = \probarg{\effectdens}{\prepdens'}.
\end{equation}
This implies
\begin{equation}
    \prepdens \indist{\rsete} \prepdens' \limplies \prepdens \indist{\asete} \prepdens'.
\end{equation}
Analogously, one can show that, for all $\effectdens,\effectdens' \in \convhull\asete$,
\begin{equation}
    \effectdens \indist{\rsets} \effectdens' \limplies \effectdens \indist{\asets} \effectdens'.
\end{equation}
Looking back at \cref{def:faithful reference}, we have thus shown that the existence of a relative noncontextual ontological model for a prepare-and-measure scenario $\apm$ with respect to a reference $\rpm$ implies that $\rpm$ is a faithful reference with respect to $\apm$.
\end{proof}

\LemInclusionRelations*
\begin{proof}
Suppose that $\apm$ admits a relative noncontextual ontological model\linebreak$\omodel$ with respect to the reference $\rpm$.
Then, the same ontological model is noncontextual with respect to the reference $\rpmm$, too. 
The only constraint to check is \cref{eq:relative noncontextuality}: for all $\prepdens,\prepdens'\in\convhull\asets$, we have both that $\prepdens \indist{\rsete'} \prepdens' \limplies \prepdens \indist{\rsete} \prepdens'$ (from \cref{eq:inclusion condition} and the \cref{def:ref preorder} of the reference preorder) as well as $\prepdens \indist{\rsete} \prepdens' \limplies \ostate{\prepdens}  = \ostate{\prepdens}$ (from \cref{eq:preparation relative noncontextuality} with respect to $\rpm$). As a result, we also have $\prepdens \indist{\rsete'} \prepdens' \limplies \ostate{\prepdens} = \ostate{\prepdens'}$, which is the noncontextuality condition of \cref{eq:preparation relative noncontextuality} with respect to the reference $\rpmm$. The ontic response function noncontextuality, \cref{eq:effect relative noncontextuality}, is established analogously.
\end{proof}

\LemPerfectStateDistinguishability*
\begin{proof}
Let the ontic space be $\ospace = \asets$. Define the ontic state and ontic response function mappings as follows: for all $\prepdens \in \convhull\asets$, $\prepdens = \sum_{\hv\in\ospace=\asets} \prob_\prepdens(\hv)\hv$, for all $\effectdens\in\convhull\asete$, for all $\hv \in \ospace = \asets$, let
\begin{subequations}
\begin{align}
    \ostate{\prepdens}(\hv) &= \prob_\prepdens(\hv), \\
    \orep{\effectdens}(\hv) &= \probarg{\effectdens}{\hv}.
\end{align}
\end{subequations}
These form well-defined distributions and effects in the sense of \cref{def:ontological model} (the representation of the trivial effect and the additivity under coarse-graining of the response function follows directly by the analogous properties of the probabilities $\probarg{\cdot}{\cdot}$ --- see \cref{sec:operational primitices}).
We can then verify \cref{eq:ontic statistics}:
\begin{subequations}
\begin{align}
    \sum_{\hv\in\ospace} \orep{\effectdens}(\hv)\ostate{\prepdens}(\hv) &= \sum_{\hv\in\ospace} \probarg{\effectdens}{\hv} \prob_\prepdens(\hv) \\
    &= \prob\left( \effectdens \middle| \textstyle\sum_{\hv\in\ospace = \asets} \prepdensarg\hv \hv \right) \\
    &= \probarg{\effectdens}{\prepdens}.
\end{align}
\end{subequations}
The relative noncontextuality of the ontic response function, \cref{eq:effect relative noncontextuality}, is satisfied, since, given $\effectdens,\effectdens' \in \convhull\asete$ such that $\effectdens \indist{\rsets} \effectdens'$, we also have, since $\rpm$ is a faithful reference for $\apm$, that $\effectdens \indist{\asets} \effectdens'$, and as such, for all $\hv \in \ospace = \asets$, 
\begin{equation}
    \effectdens \indist{\rsets} \effectdens' \limplies \orep{\effectdens}(\hv) = \probarg{\effectdens}{\hv} = \probarg{\effectdens'}{\hv} = \orep{\effectdens'}(\hv).
\end{equation}
The relative noncontextuality of the ontic state distribution, \cref{eq:preparation relative noncontextuality}, is trivially satisfied thanks to \cref{eq:perfect preparation distinguishability condition}: for all $\prepdens,\prepdens'\in\convhull\asets$, we have
\begin{equation}
    \prepdens \indist{\rsete} \prepdens' \limplies \prepdens = \prepdens' \limplies \ostate{\prepdens} = \ostate{\prepdens'}.
\end{equation}
The proposed relative noncontextual ontological model is thus a valid choice, and this concludes the proof.
\end{proof}

\newpage
\section{Relation to operationally noncontextual ontological models}
\label{sec:proof reduced space equivalence}

\subsection{Preliminaries}

Let us first define a quantum model for a prepare-and-measure scenario.
In the following, $\linops\hil$ denotes the real vector space of Hermitian linear operators acting on $\hil$, where $\hil$ is a finite-dimensional complex Hilbert space.
\begin{definition}
\label{def:quantum model}
Let $\apm$ be a prepare-and-measure scenario. A quantum model of $\apm$ is a triplet
\begin{equation}
    \qmodel
\end{equation}
where $\hil$ is a finite-dimensional Hilbert space, and the maps
\begin{subequations}
\begin{align}
    \qstate &\st \convhull\asets \to \linops{\hil}, \\
    \qeffect &\st \convhull\asete \to \linops{\hil},
\end{align}
\end{subequations}
are both convex-linear, and are such that for all $\prepdens\in\convhull\asets$, for all $\effectdens\in\convhull\asete$, it holds that
\begin{myitem}
\item $\qstate(\prepdens)$ is a density matrix, i.e., $\qstate(\prepdens)$ is positive semidefinite and trace-one;
\item $\qeffect(\effectdens)$ is a POVM element, i.e., $\qeffect(\effectdens)$ and $\id\hil-\qeffect(\effectdens)$ are positive semidefinite, and for all $\effect\in\asete$ such that $\effect = \eventset|\proc$ for some $\proc\in\procset$, it holds that
\begin{equation}
\label{eq:trivial qeffect}
    \qeffect(\effect) = \id\hil,
\end{equation}
and for all $\effect_0\in\asete$, $\{\effect_i\in\asete\}_{i=1}^n$ satisfying \cref{eq:gen add}, it holds that 
\begin{equation}
\label{eq:additivity qeffect}
\sum_{i=1}^n \qeffect(\effect_i) = \qeffect(\effect_0);
\end{equation}
\item the quantum model reproduces the statistics of the prepare-and-measure scenario:
\begin{equation}
\label{eq:qstats}
    \Tr[\qeffect(\effectdens)\qstate(\prepdens)] = \probarg{\effectdens}{\prepdens}.
\end{equation}
\end{myitem}
\end{definition}
Any prepare-and-measure scenario admits such a quantum model:%
\begin{lemma}
\label{lem:universal quantum model}
Let $\apm$ be a prepare-and-measure scenario such that $\asets = \{\prep_j\}_{j=1}^d$ for some $d\in\mathbb N$. Then, a valid quantum model is the following: let $\hil = \mathbb{C}^d$, and let $\{\ket j \}_{j=1}^d$ be an orthonormal basis of $\hil$. Define, for all $\prepdens \in \convhull\asets$, for all $\effectdens\in\convhull\asete$,
\begin{subequations}
\label{eq:trivial qmodel}
\begin{align}
    \qstate(\prepdens) &= \sum_{j=1}^d \prepdensarg{\prep_j} \ketbra{j}, \\
    \qeffect(\effectdens) &= \sum_{j=1}^d \probarg{\effectdens}{\prep_j} \ketbra{j}.    
\end{align}
\end{subequations}
\end{lemma}
\begin{proof}
These clearly satisfy the constraints that a quantum model must have. In particular, the convex-linearity and the properties of \cref{eq:trivial qeffect,eq:additivity qeffect} of the mapping $\qeffect$ follow from the analogous properties of $\probarg{\cdot}{\cdot}$. The statistics of the quantum model satisfy
\begin{subequations}
\begin{align}
    \Tr[\qeffect(\effectdens)\qstate(\prepdens)] &= \sum_{j,k = 1}^d \probarg{\effectdens}{\prep_j} \prepdensarg{\prep_k} \delta_{j,k} \\
    &= \probarg{\effectdens}{\prepdens},
\end{align}
\end{subequations}
as expected.
\end{proof}

We need to restrict the scope slightly to match the quantum prepare-and-measure scenarios of \citeref{gitton_solvable_2022}:%
\begin{definition}
\label{def:special pm scenario}
We say that a prepare-and-measure scenario $\apm$ (\cref{def:pm scenario}) is outcome-complete if for all $\effect\in\asete$ such that $\effect = \event|\proc$ for some $\event\in\sigmaalg$ and $\proc\in\procset$, there exists $\{\event_i\in\sigmaalg\}_{i=1}^n$ such that $\effect_i = \event_i|\proc \in \asete$ and such that $\event \cup (\cup_{i=1}^n \event_i) = \eventset$.
\end{definition}
We associate to a prepare-and-measure scenario a quantum version thereof through a quantum model:
\begin{definition}
\label{def:associated pm scenario}
Let $\apm$ be an outcome-complete prepare-and-measure scenario. Given a quantum model $\qmodel$, we define the quantum prepare-and-measure scenario associated to $\apm$ through
\begin{subequations}
\begin{align}
    \qsets &= \convhull{\{\qstate(\prep)\}_{\prep\in\asets}}, \label{eq:def qsets} \\
    \qsete &= \convhull{\{\qeffect(\effect)\}_{\effect\in\asete}}. \label{eq:def qsete}
\end{align}
\end{subequations}
\end{definition}
Such quantum prepare-and-measure scenarios match the setting of \citeref{gitton_solvable_2022}:%
\begin{lemma}
\label{lem:valid quantum pm scenario}
For any outcome-complete prepare-and-measure scenario $\apm$, the quantum prepare-and-measure scenario $\pmargs{\qsets}{\qsete}$ of \cref{def:associated pm scenario} satisfies definitions 1 and 2 of \citeref{gitton_solvable_2022}.
\end{lemma}
\begin{proof}
The convexity and non-emptiness are trivial. Since $\emptyset,\eventset\in\asete$ (by \cref{def:pm scenario}), we also have that $\qeffect(\emptyset),\qeffect(\eventset) \in \qsete$. \Cref{eq:trivial qeffect} implies that $\qeffect(\eventset) = \id\hil$, and \cref{eq:additivity qeffect} is easily seen to imply that $\qeffect(\emptyset) = 0_{\hil}$. This shows that $0_\hil,\id\hil \in \qsete$. It only remains to show that for all $N \in \qsete$, there exist $\{N_k\in\qsete\}_k$ such that $N + \sum_k N_k = \id\hil$, but this follows directly from \cref{def:special pm scenario} together with \cref{eq:additivity qeffect}.
\end{proof}

Before turning to the proof of \cref{th:reduced space equivalence}, we first show a useful lemma.
\begin{lemma}
\label{lem:reduced space indistinguishability}
Let $\apm$ be an outcome-complete prepare-and-measure scenario, and let $\qmodel$ be a quantum representation thereof. Let $\mainr \subseteq \linops{\hil}$ be the associated reduced space (definition 3 of \citeref{gitton_solvable_2022}), and $\proj{\mainr}{\cdot} : \linops{\hil} \to \mainr$ be the corresponding orthogonal projection. Then, the following holds: for all $\prepdens,\prepdens' \in \convhull\asets$, for all $\effectdens,\effectdens' \in \convhull\asete$, 
\begin{subequations}
\begin{align}
    \proj{\mainr}{\qstate(\prepdens)} = \proj{\mainr}{\qstate(\prepdens')} &\lequiv \prepdens \indist{\asete} \prepdens', \label{eq:reduced space indist state}\\
    \proj{\mainr}{\qeffect(\effectdens)} = \proj{\mainr}{\qeffect(\effectdens')} &\lequiv \effectdens \indist{\asets} \effectdens'. \label{eq:reduced space indist effect}
\end{align}
\end{subequations}
\end{lemma}
\begin{proof}
Suppose that $\prepdens,\prepdens' \in \convhull\asets$ are such that 
\begin{equation}
\label{eq:temp mainr 1}
\proj{\mainr}{\qstate(\prepdens)} = \proj{\mainr}{\qstate(\prepdens')}.
\end{equation}
Corollary B.6 of \citeref{gitton_solvable_2022} states that the linear span of $\proj{\mainr}{\qsete}$ is the whole $\mainr$. Thus, \cref{eq:temp mainr 1} is equivalent to: for all $\effect\in\asete$,
\begin{equation}
\label{eq:temp mainr 2}
    \scal{\proj{\mainr}{\qeffect(\effect)}}{\proj{\mainr}{\qstate(\prepdens)}}{\mainr}
    =
    \scal{\proj{\mainr}{\qeffect(\effect)}}{\proj{\mainr}{\qstate(\prepdens')}}{\mainr}.
\end{equation}
Thanks to proposition 4 of \citeref{gitton_solvable_2022}, \cref{eq:temp mainr 2} is equivalent to: for all $\effect\in\asete$,
\begin{equation}
\label{eq:temp mainr 3}
    \Tr[\qeffect(\effect)\qstate(\prepdens)]
    =
    \Tr[\qeffect(\effect)\qstate(\prepdens')].
\end{equation}
But since $\qmodel$ is a quantum model for $\apm$, we have thanks to \cref{eq:qstats} that \cref{eq:temp mainr 3} is equivalent to: for all $\effect\in\asete$,
\begin{equation}
    \probarg{\effect}{\prepdens} = \probarg{\effect}{\prepdens'},
\end{equation}
or, equivalently, $\prepdens \indist{\asete} \prepdens'$. This concludes the proof of \cref{eq:reduced space indist state}. The proof of \cref{eq:reduced space indist effect} is completely analogous: corollary B.6 of \citeref{gitton_solvable_2022} also states that the linear span of $\proj{\mainr}{\qsets}$ is the whole $\mainr$, and the rest of the argument follows through.
\end{proof}

\subsection{Proof of equivalence}

We now turn to the proof of \cref{th:reduced space equivalence}.

\begin{theorem}[Formalization of \cref{th:reduced space equivalence}]
Consider an outcome-complete (\cref{def:special pm scenario}) prepare-and-measure scenario $\apm$, an arbitrary quantum model $\qmodel$ (\cref{def:quantum model}) thereof, and the associated quantum prepare-and-measure scenario $\pmargs{\qsets}{\qsete}$ (\cref{def:associated pm scenario}). It holds that $\apm$ admits a relative noncontextual ontological model with respect to the reference $\apm$ if and only if $\pmargs{\qsets}{\qsete}$ admits a Riemann integrable (definition 17 of \citeref{gitton_solvable_2022}) operationally noncontextual ontological model (definition 6 of \citeref{gitton_solvable_2022}).
\end{theorem}

\begin{proof} ``$\Longrightarrow$''.
Suppose that $\apm$ admits a relative noncontextual ontological model with respect to the reference $\apm$. Denote this model as
\begin{equation}
\Big( \ospace, \big\{ \ostate{\prepdens} \big| \prepdens \in\convhull\asets\big\} , \big\{ \orep{\effectdens} \big| \effectdens\in\convhull\asete \big\} \Big).
\end{equation}
We will construct an operationally noncontextual ontological model for $\pmargs{\qsets}{\qsete}$ that satisfies definition 6 of \citeref{gitton_solvable_2022}. Denote the reduced space of $\pmargs{\qsets}{\qsete}$ as $\mainr$, and the operationally noncontextual ontological model as
\begin{equation}
\Big(\ospace', \big\{\ostate{\bar\sigma}' \big| \bar\sigma \in \proj{\mainr}{\qsets}\big\}, \big\{ \orep{\bar N}' \big| \bar N \in \proj{\mainr}{\qsete}\big\} \Big).
\end{equation}
Define $\ospace' = \ospace$. Then, for all $\bar\sigma\in \proj{\mainr}{\qsets}$, define $\ostate{\bar\sigma}$ as follows: choose any $\prepdens \in \convhull\asets$ such that $\bar\sigma = \proj{\mainr}{\qstate(\prepdens)}$ (thanks to \cref{eq:def qsets} and the convex-linearity of the map $\qstate$ implied by \cref{def:quantum model}, the existence of such a $\prepdens \in \convhull{\asets}$ is guaranteed), and then define $\ostate{}'$ through
\begin{equation}
    \label{eq:tentative def of quantum mu}
    \ostate{\bar\sigma}' = \ostate{\prepdens}.
\end{equation}
This definition is well-posed: suppose that there exist $\prepdens,\prepdens'\in\convhull{\asets}$ such that
\begin{equation}
    \bar \sigma = \proj{\mainr}{\qstate(\prepdens)} = \proj{\mainr}{\qstate(\prepdens')}.
\end{equation}
Then, \cref{lem:reduced space indistinguishability} implies that $\prepdens \indist{\asete} \prepdens'$, and the relative noncontextuality of the ontic state distribution (\cref{eq:preparation relative noncontextuality}) implies that $\ostate{\prepdens} = \ostate{\prepdens'}$. The normalization and nonnegativity of $\ostate{}$ transfer to $\ostate{}'$ directly. To see the convex-linearity of $\ostate{}'$, let $p\in [0,1]$ and $\bar \sigma, \bar\sigma' \in \proj{\mainr}{\qsets}$. Choose $\prepdens,\prepdens' \in \convhull\asets$ such that $\bar \sigma = \proj{\mainr}{\qstate(\prepdens)}$, $\bar\sigma' = \proj{\mainr}{\qstate(\prepdens')}$. Then, $p \bar \sigma + (1-p) \bar\sigma' \in \proj{\mainr}{\qsets}$ is such that $p \bar \sigma + (1-p) \bar\sigma' = \proj{\mainr}{\qstate(p\prepdens + (1-p)\prepdens')}$. Thus, by definition of $\ostate{}'$ and convex-linearity of $\ostate{}$, we have
\begin{equation}
    \ostate{p\bar \sigma + (1-p)\bar \sigma'}' = \ostate{p \prepdens + (1-p) \prepdens'} = p\ostate{\prepdens} + (1-p)\ostate{\prepdens'} = p\ostate{\bar\sigma}' + (1-p) \ostate{\bar\sigma'}',
\end{equation}
as expected. Turning to ontic response functions, define $\orep{}'$ analogously: for all $\bar N \in \proj{\mainr}{\qsete}$, for any $\effectdens\in\convhull\asete$ such that $\bar N = \proj{\mainr}{\qeffect(\effectdens)}$ (using \cref{eq:def qsete} and the convex-linearity of the map $\qeffect$ implied by \cref{def:quantum model}), define
\begin{equation}
\label{eq:def orep barn}
    \orep{\bar N}' = \orep{\effectdens}.
\end{equation}
This is well-posed again thanks to \cref{lem:reduced space indistinguishability} and the relative noncontextuality of the response function $\orep{}$ (if $\effectdens,\effectdens'$ are valid choices in \cref{eq:def orep barn}, then $\effectdens\indist{\asets}\effectdens'$, and then \cref{eq:effect relative noncontextuality} implies that $\orep{\effectdens} = \orep{\effectdens'}$). The range and convex-linearity of $\orep{}'$ are inherited from those of $\orep{}$, following the arguments laid out for $\ostate{}$ and $\ostate{}'$. The normalization of $\orep{}'$, initially stated as equation (2.12b) in \citeref{gitton_solvable_2022}, simplifies to equation (2.23) of \citeref{gitton_solvable_2022} that states that it is sufficient to require, for all $\hv \in \ospace$, that\footnote{The argument is laid out as part of the proof of theorem 1 of \citeref{gitton_solvable_2022} --- its essence is that one can extend $\orep{}'$ linearly.}
\begin{equation}
\label{eq:to be shown orep}
\orep{\proj{\mainr}{\id\hil}}'(\hv) = 1.
\end{equation} 
Recall that $\id\hil \in \qsete$ as shown in \cref{lem:valid quantum pm scenario}, and that $\proj{\mainr}{\qeffect(\eventset)} = \proj{\mainr}{\id\hil}$ (\cref{def:quantum model}) where $\eventset \in \asete$ (\cref{def:pm scenario}). We thus have:
\begin{equation}
    \orep{\proj{\mainr}{\id\hil}}'(\hv) = \orep{\eventset}(\hv) = 1,
\end{equation}
where the last equality follows from \cref{eq:trivial response function}. This thus establishes \cref{eq:to be shown orep}. Lastly, we need to verify that the operationally noncontextual ontological model reproduces the relevant statistics: this reads, for all $\bar\sigma\in\proj{\mainr}{\qsets}$, for all $\bar N \in \proj{\mainr}{\qsete}$, for any $\prepdens\in\convhull\asets$ and $\effectdens\in\convhull\asete$ such that $\bar\sigma = \proj{\mainr}{\qstate(\prepdens)}$ and $\bar N = \proj{\mainr}{\qeffect(\effectdens)}$,
\begin{subequations}
\begin{align}
    \sum_{\hv \in \ospace} \orep{\bar N}'(\hv) \ostate{\bar\sigma}'(\hv) 
    &= \sum_{\hv\in\ospace} \orep{\effectdens}(\hv) \ostate{\prepdens}(\hv) &\eqexpl{\cref{eq:tentative def of quantum mu,eq:def orep barn}}\\
    &= \probarg{\effectdens}{\prepdens} &\eqexpl{\cref{eq:ontic statistics}}\\
    &= \Tr[\qeffect(\effectdens)\qstate(\prepdens)] &\eqexpl{\cref{eq:qstats}}\\
    &= \scal{\proj{\mainr}{\qeffect(\effectdens)}}{\proj{\mainr}{\qstate(\prepdens)}}{\mainr} &\eqexpl{proposition 4 of \citeref{gitton_solvable_2022}}\\
    &= \scal{\bar N}{\bar\sigma}{\mainr},
\end{align}
\end{subequations}
thus reproducing the constraint of equation (2.13) of \citeref{gitton_solvable_2022}. To conclude, note that since $\ospace$ is finite, the resulting operationally noncontextual ontological model is Riemann integrable (definition 17 of \citeref{gitton_solvable_2022}).

``$\Longleftarrow$''. Let us now investigate the converse direction. Suppose that $\pmargs{\qsets}{\qsete}$ admits a Riemann integrable operationally noncontextual ontological model that we denote
\begin{equation}
\Big(\ospace', \big\{\ostate{\bar\sigma}' \big| \bar\sigma \in \proj{\mainr}{\qsets}\big\}, \big\{ \orep{\bar N}' \big| \bar N \in \proj{\mainr}{\qsete}\big\} \Big).
\end{equation}
Thanks to theorem 3 of \citeref{gitton_solvable_2022}, we can assume that $\ospace'$ is a finite set. Indeed, theorem 3 of \citeref{gitton_solvable_2022} implies that, without loss of generality, $|\ospace'| \leq \dim(\mainr)^2$, and this upper bound is finite: indeed, $\mainr\subseteq\linops{\hil}$ and $\hil$ can be taken to be a finite-dimensional Hilbert space as shown in \cref{lem:universal quantum model} (which ultimately relies on the finiteness of the set $\asets$ as posited in \cref{def:pm scenario}). Let us now define a relative noncontextual ontological model for $\apm$ with respect to the reference $\apm$. This model is denoted
\begin{equation}
\Big( \ospace, \big\{ \ostate{\prepdens} \big| \prepdens \in\convhull\asets\big\} , \big\{ \orep{\effectdens} \big| \effectdens\in\convhull\asete \big\} \Big).
\end{equation}
We let $\ospace = \ospace'$ (a valid choice, since $\ospace'$ is finite). Then, for all $\prepdens\in\convhull\asets$, for all $\effectdens\in\convhull\asete$, for all $\hv\in\ospace$, define
\begin{subequations}
\begin{align}
    \ostate{\prepdens}(\hv) &= \ostate{\proj{\mainr}{\qstate(\prepdens)}}'(\hv), \label{eq:def ostate from ostate'} \\
    \orep{\effectdens}(\hv) &= \ostate{\proj{\mainr}{\qeffect(\effectdens)}}'(\hv). \label{eq:def orep from orep'}
\end{align}
\end{subequations}
The nonnegativity and convex-linearity of $\ostate{}$ and $\orep{}$ are inherited directly from $\ostate{}'$ and $\orep{}'$ (using also the linearity of the maps $\qstate$ and $\qeffect$). The normalization of $\ostate{}$ is inherited from that of $\ostate{}'$. To investigate the normalization of $\orep{}$, consider an arbitrary $\effect\in\asete$ such that $\effect = \eventset|\proc$ for some $\proc\in\procset$. Then, \cref{def:quantum model} implies that $\qeffect(\effect) = \id\hil$. This implies that for all $\hv\in\ospace$,
\begin{equation}
    \orep{\effect}(\hv) = \orep{\proj{\mainr}{\qeffect(\effect)}}'(\hv) = \orep{\proj{\mainr}{\id\hil}}'(\hv) = 1,
\end{equation}
where the last equality follows from the normalization of $\orep{}'$ (equation (2.12b) of \citeref{gitton_solvable_2022}). Thus, the constraint of \cref{eq:trivial response function} in the relative noncontextual ontological model is satisfied. Consider now the constraint of \cref{eq:ontic response function additivity}: for all $\effect_0\in\asete$, $\{\effect_i\in\asete\}_{i=1}^n$ satisfying \cref{eq:gen add}, 
\begin{subequations}
\begin{align}
    \sum_{i=1}^n \orep{\effect_i}(\hv) 
    &= \sum_{i=1}^n \orep{\proj{\mainr}{\qeffect(\effect_i)}}'(\hv) \\
    &= \orep{\proj{\mainr}{\sum_{i=1}^n \qeffect(\effect_i)}}'(\hv) &\eqexpl{$\orep{}'$ uniquely linearly extendable (prop.~9 of \citeref{gitton_solvable_2022})}\\
    &= \orep{\proj{\mainr}{\qeffect(\effect_0)}}'(\hv) &\eqexpl{\cref{eq:additivity qeffect}}\\
    &= \orep{\effect_0}(\hv),
\end{align}
\end{subequations}
so that \cref{eq:ontic response function additivity} is satisfied. \Cref{eq:ontic statistics} is easy to verify: for all $\prepdens\in\convhull\asets$, for all $\effectdens\in\convhull\asete$,
\begin{subequations}
\begin{align}
    \sum_{\hv\in\ospace} \orep{\effectdens}(\hv) \ostate{\prepdens}(\hv) 
    &= \sum_{\hv\in\ospace} \orep{\proj{\mainr}{\qeffect(\effectdens)}}'(\hv) \ostate{\proj{\mainr}{\qstate(\prepdens)}}'(\hv) \\
    &= \scal{\proj{\mainr}{\qeffect(\effectdens)}}{\proj{\mainr}{\qstate(\prepdens)}}{\mainr} &\eqexpl{equation (2.13) of \citeref{gitton_solvable_2022}} \\
    &= \Tr[\qeffect(\effectdens)\qstate(\prepdens)] &\eqexpl{proposition 4 of \citeref{gitton_solvable_2022}}\\
    &= \probarg{\effectdens}{\prepdens}. &\eqexpl{\cref{eq:qstats}}
\end{align}
\end{subequations}
Finally, let us now consider the relative noncontextuality constraints of \cref{eq:preparation relative noncontextuality,eq:effect relative noncontextuality}: suppose that $\prepdens,\prepdens'\in\convhull\asets$ are such that $\prepdens\indist{\asete} \prepdens'$. Then, \cref{lem:reduced space indistinguishability} implies that $\proj{\mainr}{\qstate(\prepdens)} = \proj{\mainr}{\qstate(\prepdens')}$, so that \cref{eq:def ostate from ostate'} implies $\ostate{\prepdens} = \ostate{\prepdens'}$. This establishes \cref{eq:preparation relative noncontextuality}, and the case of \cref{eq:effect relative noncontextuality} is analogous.
\end{proof}

\end{document}

%% file: venn_diagram.tex
\newcommand{\vennr}{40pt}
$$
\centertikz{
\begin{scope}[line width=1.7pt,draw=red]
\draw[draw=red,dashed,fill=red,fill opacity=0.08] (0,0) circle[x radius=3*\vennr, y radius=1.2*\vennr];
\draw[fill=red,fill opacity=0.08] (0,0) circle[x radius=2*\vennr, y radius=1.15*\vennr,rotate=50];
\draw[fill=red,fill opacity=0.08] (0,0) circle[x radius=2*\vennr, y radius=1.15*\vennr,rotate=130];
\node[circle,draw=red,inner color=blue!15,outer color=blue!30,minimum size=\vennr] (pm) at (0,0) {PM scenario};
\end{scope}
\begin{scope}[yshift=-2.8*\vennr]
\node[notion] (op) at (-2.5*\vennr,0) {Operational NC};
\node[notion] (pr) at (0*\vennr,0) {Pragmatic NC};
\node[notion] (ontic) at (2.5*\vennr,0) {Ontic-system NC};
\end{scope}
\begin{scope}[line width=1.1pt, arrows={->},draw=black]
\draw (op) -- (pm);
\draw (pr) -- (-0.5*\vennr,-1.7*\vennr);
\node at (0,-2.1*\vennr) {\small or};
\draw (pr) -- ( 0.5*\vennr,-1.7*\vennr);
\draw (ontic) -- node[right] {?} (2.2*\vennr,-0.85*\vennr);
\end{scope}
%
%
%
}
$$

%% file: mind_map.tex
\newcommand{\mindmapspace}{20pt}
\newcommand{\mindmapspaceiii}{25.5pt}
\newcommand{\mindmapvspace}{1cm}
$$
\centertikz{
\begin{scope}[
xscale=3.8,
yscale=2.5,
concept/.style={notion},
derives/.style={line width=1.2pt,draw=black,arrows={->}},
]
%
%
\node[notion] (preorder) at (0,-2) {$\begin{gathered}\text{Reference preorder} \\\text{(\cref{sec:reference preorder})}\end{gathered}$};
\node[notion] (indist) [above=\mindmapvspace of preorder] {$\begin{gathered}\text{Indistinguishability}\\\text{(\cref{sec:introducing indistinguishability})}\end{gathered}$};
\node[notion] (rncom) [right=\mindmapspaceiii of preorder] {$\begin{gathered}\text{Relative noncontextual} \\\text{OMs (\cref{sec:definitions om})}\end{gathered}$};
\draw[derives] (indist) -- (rncom.north west);
\draw[derives] (indist) -- (preorder);
\node[notion] (pm) [above=1cm of rncom] {$\begin{gathered}\text{PM scenarios}\\\text{(\cref{sec:operational primitices})}\end{gathered}$};
\draw[derives] (pm) -- (rncom);
\draw[derives] (pm) -- (preorder.north east);
\node[concept] (pragm) [right=\mindmapspaceiii of rncom] {$\begin{gathered}\text{Pragmatic}\\\text{system}\end{gathered}$};
\node[concept] (ontic) [right=\mindmapspaceiii of pragm] {$\begin{gathered}\text{Ontic}\\\text{system}\end{gathered}$};
\node[notion] (graph) [below=\mindmapvspace of preorder] {$\begin{gathered}\text{Relative NC graph}\\\text{(\cref{sec:relative nc graph})}\end{gathered}$};
\draw[derives] (preorder) -- (graph);
\draw[derives] (rncom) -- (graph.north east);
\node[notion] (opernc) [right=\mindmapspace of graph] {$\begin{gathered}\text{Operational NC}\\\text{(\cref{sec:operational noncontextuality})}\end{gathered}$};
\draw[derives] (rncom) -- (opernc);
\node[concept] (pncom) [right=\mindmapspace of opernc] {$\begin{gathered}\text{Pragmatic NC}\\\text{(\cref{sec:pragmatic noncontextuality})}\end{gathered}$};
\draw[derives] (rncom) -- (pncom.north west);
\draw[derives] (pragm) -- (pncom);
\node[concept] (oncom) [right=\mindmapspace of pncom] {$\begin{gathered}\text{Ontic-system NC}\\\text{(\cref{sec:ontic-system noncontextuality})}\end{gathered}$};
\draw[derives] (rncom) -- (oncom.north west);
\draw[derives] (ontic) -- (oncom);
\end{scope}
}
$$

%% file: typical_graph.tex
$$
\centertikz{{[yscale=1.05]
\node[nclnode] (a1) at (0,0) {};
%
\node[nclnode] (b1) at (-1,1) {};
\node[nclnode] (b2) at ( 1,1) {};
\draw[preorder] (a1) -- (b1);
\draw[preorder] (a1) -- (b2);
%
\node[nclnode] (c1) at (-2,2) {};
\node[ clnode] (c2) at (-1,2) {};
\node[nclnode] (c4) at ( 1,2) {};
\node[ clnode] (c5) at ( 2,2) {};
\draw[preorder] (b1) -- (c1);
\draw[preorder] (b1) -- (c2);
\draw[preorder] (b1) -- (c4);
\draw[preorder] (b2) -- (c4);
\draw[preorder] (b2) -- (c5);
%
\node[nclnode] (d1) at (-2,3) {};
\node[ clnode] (d2) at (-1,3) {};
\node[ clnode] (d3) at ( 0,3) {};
\node[nclnode] (d4) at ( 1,3) {};
\node[ clnode] (d5) at ( 2,3) {};
\draw[preorder] (c1) -- (d1);
\draw[preorder] (c2) -- (d2);
\draw[preorder] (c2) -- (d3);
\draw[preorder] (c4) -- (d4);
\draw[preorder] (c4) -- (d3);
\draw[preorder] (c5) -- (d5);
%
\node[nclnode] (e1) at (-3,4) {};
\node[ clnode] (e2) at (-2,4) {};
\node[ clnode] (e3) at (-1,4) {};
\node[ clnode] (e4) at ( 0,4) {};
\node[ clnode] (e5) at ( 1,4) {};
\node[nclnode] (e6) at ( 2,4) {};
\node[ clnode] (e7) at ( 3,4) {};
\draw[preorder] (d1) -- (e1);
\draw[preorder] (d1) -- (e2);
\draw[preorder] (d2) -- (e2);
\draw[preorder] (d2) -- (e3);
\draw[preorder] (d3) -- (e4);
\draw[preorder] (d4) -- (e5);
\draw[preorder] (d4) -- (e6);
\draw[preorder] (d5) -- (e7);
%
\node[ clnode] (f1) at (-3,5) {};
\node[nclnode] (f2) at (-2,5) {};
\node[ clnode] (f3) at (-1,5) {};
\node[ clnode] (f4) at ( 0,5) {};
\node[ clnode] (f5) at ( 1,5) {};
\node[nclnode] (f6) at ( 2,5) {};
\node[ clnode] (f7) at ( 3,5) {};
\draw[preorder] (e1) -- (f1);
\draw[preorder] (e1) -- (f2);
\draw[preorder] (e2) -- (f3);
\draw[preorder] (e3) -- (f3);
\draw[preorder] (e3) -- (f4);
\draw[preorder] (e4) -- (f4);
\draw[preorder] (e5) -- (f5);
\draw[preorder] (e6) -- (f6);
\draw[preorder] (e7) -- (f7);
%
\node[ clnode] (g1) at (-2,6) {};
\node[ clnode] (g2) at (-1,6) {};
\node[ clnode] (g3) at ( 0,6) {};
\node[nclnode] (g4) at ( 1,6) {};
\node[ clnode] (g5) at ( 2,6) {};
\draw[preorder] (f1) -- (g1);
\draw[preorder] (f2) -- (g1);
\draw[preorder] (f3) -- (g2);
\draw[preorder] (f4) -- (g3);
\draw[preorder] (f5) -- (g3);
\draw[preorder] (f6) -- (g4);
\draw[preorder] (f7) -- (g5);
%
\node[clnode] (h1) at (-1,7) {};
\node[clnode] (h2) at ( 1,7) {};
\foreach \n/\m in {1/1,2/1,3/1,4/2,5/2}
    \draw[preorder] (g\n) -- (h\m);
\node (apm) at (2,0) {$\apm$};
\node (note) at (0,8) {\scriptsize(assumed to satisfy the assumptions of \cref{lem:perfect state distinguishability})};
{[dashed,black!90,line width=0.7pt, arrows={-Latex}]
\draw (apm) -- (a1);
\draw (note) -- (h1);
\draw (note) -- (h2);
}
}}
$$

%% file: elevator.tex
$$
\centertikz{
\begin{scope}[line width=1.5pt,rounded corners=1pt,scale=1.7]
\fill[white] (-1,0) circle (1pt);
\begin{scope}[draw=green!90!black]
\draw (2,-0.2) to [out=20,in=160] ++(3,0);
\foreach \x/\y in {0/1,1/2,2/2.5,3/2.5,4/2,5/1} 
    \draw (2.3+\x*0.4,-0.4) -- ++(0.2+\y*0.06,0.15+\y*0.08);
\end{scope}
\foreach \x in {0,3} {
\begin{scope}[xshift=\x cm]
    \fill[white] (0,0) rectangle +(1,1);
    \draw (0,0) rectangle +(1,1);
    \draw[black!80] (0.4,0.2) -- (0.6,0.2) -- (0.53,0.25) -- (0.53,0.7) -- (0.75,0.75) -- (0.25,0.65) -- (0.47,0.7) -- (0.47,0.25) -- cycle;
    \draw[fill=black!80] (0.5,0.7) circle (1.1pt);
    \draw[black!50,line width=1pt] (0.75,0.75) -- (0.85,0.55) -- (0.65,0.55) -- cycle;
    \draw[black!50,line width=1pt] (0.25,0.65) -- (0.35,0.45) -- (0.15,0.45) -- cycle;
    \draw[red,fill=red] (0.75,0.58) circle (1pt);
    \draw[blue,fill=blue] (0.25,0.48) circle (1pt);
\end{scope}
}
\draw[red] 
(0.35,0) to[out=230,in=90] 
(0.3,-0.4) to[out=70,in=230] 
(0.4,-0.2) to[out=260,in=110] 
(0.5,-0.5) to[out=70,in=280] 
(0.6,-0.2) to[out=310,in=110] 
(0.7,-0.4) to[out=90,in=310] 
(0.65,0);
\draw[black!50,line width=1pt] (0.35,-0.8) -- ++(0,0.3) (0.5,-0.85) -- ++(0,0.3) (0.65,-0.8) -- ++(0,0.3);
\draw[red,decorate,decoration={brace,amplitude=5pt}] (-0.5,0) -- node[left=5pt,black] {$\rsete^{(1)}$} (-0.5,1);
\draw[red,decorate,decoration={brace,amplitude=5pt}] (-1.2,-0.8) -- node[left,black] {$\rsete^{(2)}$} (-1.2,1);
\fill[white] (5.8,0) circle (1pt);
\node at (0.5,1.5) {$\prep^\text{(i)}$};
\node at (3.5,1.5) {$\prep^\text{(ii)}$};
\end{scope}
}
$$

%% file: micro_cut.tex
$$
\centertikz{
\begin{scope}[yscale=1.5,xscale=3.5]
%
%
\node (a) at (0.5,1) {Alice};
\foreach \xs/\ys/\xshi in {0.4/0.6/3cm} {
\begin{scope}[xshift=\xshi,yshift=0.8cm,xscale=\xs,yscale=\ys]
    \draw[line width=1.2pt,draw=red,dashed,fill=red!20] (0,0) to[out=90,in=230] (1,1) to[out=50,in=80] (2,0.5) to[out=260,in=20] (0.5,-\xs) to[out=200,in=270] cycle;
\end{scope}
}
\node at (3.8,0.8) {$R$};
\node[rectangle,draw=black!70,line width=1pt,rounded corners=2pt,inner color=black!0,outer color=black!3,minimum size=25pt] (d) at (1.5,1) {Classical device};
\draw[line width=1pt, draw=black!80, arrows={->}] (a) -- (d);
\draw[line width=1pt, draw=black!80, arrows={->}] (d) to[out=20,in=160] (3.6,1.2);
\draw[line width=1pt, draw=black!80, arrows={->}] (d) to[out=0,in=180] (3.4,1);
\draw[line width=1pt, draw=black!80, arrows={->}] (d) to[out=-20,in=200] (3.2,0.8);
\draw[line width=1.2pt, draw=red, dashed] (2.9,0) to[out=80,in=220] (3.2,2);
\draw[line width=1.2pt, draw=red, dashed] (2.6,0) to[out=60,in=230] (2.8,2);
\draw[line width=1.2pt, draw=red, dashed] (2.4,0) to[out=70,in=250] (2.5,2);
\node at (2.8,2.2) {Other possible cuts};
\end{scope}
}
$$